\documentclass[reqno]{amsart}

\usepackage{amssymb}
\usepackage{enumerate}
\usepackage[usenames,dvipsnames]{color}
\usepackage[bookmarks,colorlinks,breaklinks]{hyperref}  
\usepackage{doi,url}

\definecolor{dullmagenta}{rgb}{0.4,0,0.4}   
\definecolor{darkblue}{rgb}{0,0,0.4}
\hypersetup{linkcolor=red,citecolor=blue,filecolor=dullmagenta,urlcolor=darkblue}

\newcommand{\bla}{\color{black}}

\usepackage[notref,notcite,final]{showkeys}
\newcommand{\eq}[1]{\eqref{#1}}

\theoremstyle{remark}
\theoremstyle{definition}
\theoremstyle{plain}

\newtheorem{theorem}{Theorem}[section]
\newtheorem{proposition}[theorem]{Proposition}
\newtheorem{lemma}[theorem]{Lemma}
\newtheorem{corollary}[theorem]{Corollary}

\newtheorem{definition}[theorem]{Definition}
\newtheorem{remark}[theorem]{Remark}

\newtheorem*{remark*}{Remark}

\numberwithin{equation}{section}

\DeclareMathOperator{\supp}{supp}
\DeclareMathOperator{\tr}{tr}

\DeclareMathOperator{\dist}{dist}

\newcommand{\pr}{\prime}

\newcommand\R{\mathbb R}

\newcommand\N{\mathbb N}
\newcommand\C{\mathbb C}
\newcommand\Z{\mathbb Z}

\renewcommand\P{\mathbb P}
\newcommand\E{\mathbb E}

\renewcommand\L{\mathrm{L}}
\newcommand\bL{\boldsymbol{L}}

\newcommand{\cA}{\mathcal{A}}
\newcommand{\cB}{\mathcal{B}}
\newcommand{\cC}{\mathcal{C}}

\newcommand{\cE}{\mathcal{E}}
\newcommand{\cF}{\mathcal{F}}

\newcommand{\I}{\mathcal{I}}

\newcommand{\cJ}{\mathcal{J}}

\newcommand{\cL}{\mathcal{L}}
\newcommand{\cM}{\mathcal{M}}

\newcommand{\cP}{\mathcal{P}}

\newcommand{\cS}{\mathcal{S}}
\newcommand{\cT}{\mathcal{T}}
\newcommand{\cU}{\mathcal{U}}
\newcommand{\cW}{\mathcal{W}}

\newcommand\Chi{\raisebox{.2ex}{$\chi$}}

\newcommand{\abs}[1]{\left\lvert #1 \right\rvert}
\newcommand{\norm}[1]{\left\lVert #1 \right\rVert}
\newcommand{\scal}[1]{\left\langle #1 \right\rangle}
\newcommand{\set}[1]{\left\{ #1 \right\}}
\newcommand{\pa}[1]{\left( #1 \right)}

\newcommand{\up}[1]{^{(#1)}}

\newcommand{\x}{\mathbf{x}}
\newcommand{\bolda}{\mathbf{a}}
\newcommand{\boldx}{\mathbf{x}}
\newcommand{\boldb}{\mathbf{b}}
\newcommand{\boldu}{\mathbf{u}}
\newcommand{\boldv}{\mathbf{v}}
\newcommand{\y}{\mathbf{y}}
\newcommand{\boldy}{\mathbf{y}}

\newcommand{\boldk}{\mathbf{k}}

\newcommand{\bom}{\boldsymbol{\omega}}

\newcommand{\blam}{\boldsymbol{\Lambda}}

\newcommand{\bTheta}{\boldsymbol{\Theta}}

\newcommand{\om}{\omega}
\newcommand\eps{\varepsilon}

\newcommand\La{\Lambda}

\newcommand{\boldlambda}{\mathbf{\Lambda}}

\newcommand{\suitc}{\Xi_{L,\ell}}

\newcommand{\ml}{m_\ell}

\newcommand{\Bl}{\Bigl}
\newcommand{\Br}{\Bigr}

\newcommand{\bX}{\textbf{X}}

\newcommand\beq{\begin{equation}}
\newcommand\eeq{\end{equation}}

\newcommand{\qtx}[1]{\quad\text{#1}\quad}
\newcommand{\mqtx}[1]{\;\;\text{#1}\;\;}
\newcommand{\sqtx}[1]{\;\text{#1}\;}

\newcommand{\paS}{{S}}
\newcommand{\pas}{{S}}

\newcommand{\pat}{(2)}
\newcommand{\PET}{(\theta, E)}
\newcommand{\alphaL}{\frac{L}{3}}
\newcommand{\alphal}{\frac{\ell}{3}}

\theoremstyle{remark}

\newcommand{\be}{\begin{equation}}
\newcommand{\ee}{\end{equation}}
\newcommand{\ba}{\begin{array}}
\newcommand{\ea}{\end{array}}
\newcommand{\bal}{\begin{align}}
\newcommand{\eal}{\end{align}}
\newcommand{\bea}{\begin{eqnarray}}
\newcommand{\eea}{\end{eqnarray}}
\newcommand{\bee}{\begin{eqnarray*}}
\newcommand{\eee}{\end{eqnarray*}}

\renewcommand{\L}{\Lambda}
\newcommand{\angles}[1]{\langle #1 \rangle}

\renewcommand{\*}{{\sharp}}

\newcommand{\Ccinf}{ C_{c}^{\infty}   }
\newcommand{\Ccinfp}{ C_{c, +}^{\infty}   }

\newcommand{\vertiii}[1]{{\left\vert\kern-0.25ex\left\vert\kern-0.25ex\left\vert #1
    \right\vert\kern-0.25ex\right\vert\kern-0.25ex\right\vert}}

\begin{document}

\title[Characterization of the metal-insulator transport transition]{Characterization of the metal-insulator transport  transition for the two-particle Anderson model}

\author{Abel Klein}
\address[A. Klein]{University of California, Irvine;
Department of Mathematics;
Irvine, CA 92697-3875,  USA}
 \email{aklein@uci.edu}

 \author{Son T. Nguyen}
\address[S. T. Nguyen]{Coastline Community College,  Newport Beach, CA 92663, USA}
\email{sondgnguyen1@gmail.com}

 \author{Constanza Rojas-Molina}
\address[C. Rojas-Molina]{Institut f\"ur Angewandte Mathematik, Rheinische Friedrich-Wilhelms-Universit\"at,
Endenicher Allee 60, 53115 Bonn, Germany}
 \email{crojasm@uni-bonn.de}

\thanks{A.K. was  supported in part by the NSF under grant DMS-1001509.}
\thanks{C.R-M. was supported in part by the European Community FP7 Programme under grant agreement
number 329458 and by the Collaborative Research Center 1060 from the German Research Foundation. C.R-M. acknowledges the support and hospitality of the Isaac Newton Institute for Mathematical Sciences during the programme Periodic and Ergodic Spectral Problems.}

\date{Version of \today}

 \begin{abstract}
We extend to the two-particle Anderson model the characterization of the metal-insulator transport transition obtained in the one-particle setting by Germinet and Klein.
We show that, for any fixed number of particles, the slow spreading of wave packets in time implies the initial estimate of a modified version of the Bootstrap Multiscale Analysis.
In this new version, operators are restricted to boxes defined with respect to the pseudo-distance in which we have the slow spreading.  At the bottom of the spectrum, within the regime of one-particle dynamical localization,
we show that this modified multiscale analysis yields dynamical localization for the two-particle Anderson model, allowing us to obtain  a characterization of the metal-insulator  transport  transition for the two-particle Anderson model at the bottom of the spectrum.
\end{abstract}

\maketitle

 \tableofcontents

\section{Introduction}\label{S:intro}

Localization, or absence of transport, is considered to be a characteristic feature of random media. In the well known
one-particle Anderson model, it is known to appear at either high disorder or low energies. In recent years the multi-particle Anderson model has attracted great interest, as it is expected that localization persists in the presence of inter-particle interactions. This has been shown to be the case for short-range interactions, using either a multiscale analysis (MSA) or the fractional moment method \cite{CS1,CS2,CS3,AWmp, BCSS, CBS,KlN1,KlN2}. More recently, interactions of exponential decay and even fast polynomial decay have been considered in \cite{FW}.

In this paper we contribute to the efforts to understand the regime of localization in the multi-particle setting by extending the work of Germinet and Klein \cite{GKduke} on the characterization of the metal-insulator transport transition for the Anderson model. In the aforementioned work, the authors  used transport exponents to measure  the spreading of wave packets,  spliting the spectrum of the operator into two complementary regions: the \emph{metallic region}, where non trivial transport occurs, and the \emph{insulator region}, where transport is suppressed by dynamical localization. Germinet and Klein showed that the insulator region, where the transport exponent is null, is equivalent to the set of energies where the MSA can be applied, making this method their main tool of analysis. Moreover, they gave a lower bound for the transport exponent in the metallic region, which implies that the mobility edge, i.e., the energy that separates the regions of localization and delocalization, is a point of discontinuity for the transport exponent.

 In our work we follow the approach of  \cite{GKduke}, where the proof consists of two parts: on one hand, slow transport, identified as slow growth of the time-averaged spreading of wave packets, implies the starting hypothesis of the MSA; on the other hand, the MSA implies dynamical localization and therefore null transport. Therefore, if there is spreading of wave packets, this must be at a rate that is above a minimal amount. In the one-particle setting, this characterization proved to be crucial in the study of the Landau Hamiltonian perturbed by a random potential \cite{GKS}, one of the few models where the Anderson transport transition has been rigorously proved, together with the Anderson model in the  Bethe lattice \cite{Klmrl,KlBethe,ASW,FHS,AWres}.

In the multi-particle setting dynamical localization is obtained  with respect to the Hausdorff pseudo-distance in the multi-particle space, but the  MSA (or fractional moment method) is based on the usual boxes in the multi-particle space, i.e., boxes defined by the norm given by the maximum of the absolute value of the coordinates \cite{CS1,CS2,CS3,AWmp, BCSS, CBS,KlN1,KlN2}.  The initial step for the MSA is a statement about the finite-volume restrictions of the random Schr\"odinger operators to usual boxes, but the statement of dynamical localization is only with respect to the Hausdorff pseudo-distance.  For two particles, i.e., $N=2$, the Hausdorff and symmetrized pseudo-distances are the same, so the  statement of dynamical localization can be seen as being with respect to the symmetrized pseudo-distance.

In order to extend the characterization of Germinet and Klein to the multi-particle setting, we first show that the slow spreading of the $N$-particle wave-packets for large times, with respect to either the Hausdorff or symmetrized pseudo-distance, implies the initial step of a modified MSA in which the finite-volume restrictions of operators are defined using boxes with respect to the relevant pseudo-distance. This  holds for any fixed number of particles $N$, with either the Hausdorff or symmetrized pseudo-distance, since the Wegner estimate holds on boxes defined by either pseudo-distance.

To obtain the characterization of the metal-insulator transition as in \cite{GKduke}, we need to perform a variant of the  Bootstrap MSA (see \cite{KlN1} for the Bootstrap MSA for multi-particle Anderson models) using boxes defined by the pseudo-metric. The deterministic part of the MSA can be done for boxes defined by the symmetrized pseudo-distance, but the probabilistic part of this modified MSA would require a Wegner estimate between boxes that are far apart in the symmetrized pseudo-distance. For the multi-particle Anderson model where the single-site probability distribution is only assumed to have a bounded density with compact support such a Wegner estimate is only known between boxes that are far apart in the Hausdorff distance, and hence we only have it for the symmetrized distance if $N=2$. For this reason we restrict ourselves to two particles, for which we prove dynamical localization with respect to the symmetrized pseudo-distance, with the initial step for the MSA being a statement about the finite-volume restrictions of the random  operators to symmetrized boxes. Thus for two particles we obtain a characterization of the metal-insulator transition in the spirit of \cite{GKduke} for quantities defined with respect to the symmetrized pseudo-distance.

Recently Chulaevsky obtained a Wegner estimate between boxes that are far apart in the symmetrized distance for all $N$, for a certain class of single site probability distributions \cite{Ch2} (see also \cite{Ch1}).  Under Chulaevsky's assumptions our modified Bootstrap MSA for boxes in the symmetrized   pseudo-distance can be performed for $N$-particles, where $N$ is fixed but arbitrary, and our characterization of the metal-insulator transition  holds for $N$-particles.

The article is organized as follows: in the next section we set the notation and state the main results. In Section \ref{S:proof1} we prove that slow transport implies the initial step of a modified Bootstrap MSA. Section~\ref{S:MSA} is devoted to the Bootstrap Multiscale Analysis for symmetrized two-particle boxes.   In Subsection \ref{S:preMSA} we explain the modifications to the $N$-particle MSA of \cite{KlN1,KlN2} that are necessary in our setting,  and give details on how this modifications are implemented in Subsection \ref{S:MSA33}. In Section \ref{S:dynloc} we show that whenever this MSA can be performed in an $N$-particle setting, it yields dynamical localization with respect to the symmetrized pseudo-distance, which completes the proof of the main result. At the same time, we are able to improve the conclusions of  \cite[Corollary 1.7]{KlN1} and \cite[Theorem 1.2]{KlN2} by removing the dependency on the initial position of the particles in the statement of dynamical localization. Appendix~\ref{apW} discusses Wegner estimates for multi-particle rectangles and boxes defined by  the relevant pseudo-distances. In Appendix~\ref{apCT} we state and prove a Combes-Thomas estimate for restrictions of discrete Schr\"odinger operators to arbitrary subsets.   In  Appendix~\ref{appdiscrete}  we comment on auxiliary results known in the one-particle case, that are generalized to the multi-particle setting and an arbitrary pseudo-distance.

\section{Main definitions and results}
We start by defining the $n$-particle  Anderson model.

 \begin{definition}\label{defAndmodel}
The $n$-particle Anderson model is  the  random Schr\"odinger
operator on
$\ell^{2}(\mathbb{Z}^{nd})$ given by
\beq \label{AndH}
H_{\bom}^{(n)}: =  -\Delta^{(n)} + \,V_{\bom}^{(n)} + U ,
\eeq
where:
\begin{enumerate}
\item
$\Delta^{(n)}$ is the discrete $nd$-dimensional  Laplacian operator.
\item
$\bom=\{ \omega_x \}_{x\in
\Z^d}$ is a family of independent
identically distributed random
variables  whose  common probability
distribution $\mu$ has a bounded density $\rho$ and satisfies $\set{0, M_+} \subset \supp \mu \subseteq [0, M_+]$ for some $M_+ > 0$.
\item
$V_{\bom}^{(n)}$ is the random potential  given by
\beq
V_{\bom}^{(n)}(\x)= \sum_{i = 1, ..., n} V_{\bom}^{(1)}(x_i)\qtx{for} \x =(x_1, ..., x_n) \in \mathbb{Z}^{nd} ,
\eeq
where $V_{\bom}^{(1)}(x)=\omega_{x}$ for every $x \in \mathbb{Z}^{d}$.
\item
$U$ is a potential governing the short range interaction between the $n$ particles.  We take
\beq\label{2body}
U(\x) = \sum_{1 \leq i < j \leq n } \widetilde{U}(x_i - x_j) ,
\eeq
where $\widetilde{U}\colon \Z^d  \to [0, \infty)$, $\widetilde{U}(y)= \widetilde{U}(-y)$, and  $\widetilde{U}(y) = 0$ for $\norm{y}_\infty > r_0$ for some  $r_0 \in [1, \infty)$. \end{enumerate}
\end{definition}

The $n$-particle Anderson model is a $\Z^d$-ergodic random Schr\"odinger operator. It follows from standard arguments (cf. \cite[Proposition V.2.4]{CL}) that there exists a bounded set $\Sigma^{(n)}\subset \R$ such that $\sigma(H_{\bom}^{(n)} )=\Sigma^{(n)}$ almost-surely, where $\sigma(H)$ denotes the spectrum  of the operator $H$.
Since $\sigma(-\Delta^{(n)})=[0,4nd]$, and  both $V^{(n)}_{\bom}$ and $U$ are non-negative, we have $\Sigma^{(n)} \subset[0,+\infty)$.

 We will generally omit ${\bom}$ from the notation, and
use the following notation and definitions:
\begin{enumerate}

\item  Given $x = (x_1, \ldots, x_d) \in \R^{d}$, we set $\norm{x}=\norm{x}_\infty  := \max \{\abs{x_1}, \ldots, \abs{x_d} \}$.  If  $\bolda = (a_1, \ldots, a_n) \in \R^{nd}$, we let
  $\norm{\bolda} := \max \{ \norm{a_1}, \ldots, \norm{a_n} \}$,   $\cS_{ \bolda} = \bigl \{a_1, \,...,\, a_n   \bigr\}$, and  $\scal{\bolda} := \scal{\norm{\bolda}}$, where $\scal{t}:= \sqrt{ 1+ t^{2}}$ for $t\ge 0$.

\item For $n \in \N$, we set $\cP_n$ to be the set of all permutations of $n$ elements. Moreover, if $\x = \pa{x_1,...,x_n} \in \R^{nd}$ and $\pi \in \cP_n$, we write
$
\pi(\x) = \pa{x_{\pi(1)}, \dots, x_{\pi(n)}} $.

\item  We introduce  one distance and two pseudo-distances (which we will incorrectly call distances) in $\R^{nd}$.  They are defined as follows for $\bolda, \boldb  \in \R^{nd}$:
\begin{enumerate}

\item  The norm distance: \quad
$
\dist_\infty ( \bolda, \, \boldb)=\dist ( \bolda, \, \boldb):=\norm{\bolda - \boldb}
$.

\item  The symmetrized distance:
\beq
\dist_{S} ( \bolda, \, \boldb):= \min_{\pi \in \cS_n} \set{\norm{\pi(\bolda) - \boldb}} = \min_{\pi \in \cS_n} \set{\norm{\pi(\boldb) - \bolda}} .
\eeq

\item The Hausdorff distance:
\begin{align}\label{Hausddist}
\dist_{H} ( \bolda, \, \boldb)&:=  \dist_{H}(\cS_{ \bolda},\cS_{ \boldb})\\ \notag & := \max \Bl\{ \max_{x \in \cS_{ \bolda}} \, \min_{y \in \cS_{ \boldb}} \norm{x - y} \, , \,   \max_{y \in \cS_{ \boldb}} \, \min_{x \in \cS_{ \bolda}} \norm{x - y}  \Br\}\\
&= \max \Bl\{ \max_{x \in \cS_{ \bolda}} \, \dist(x, \,\,\cS_{ \boldb}) \, , \,   \max_{y \in \cS_{ \boldb}} \, \dist(y, \,\,\cS_{ \bolda})  \Br\}.
\notag
\end{align}

\end{enumerate}

We will write $\dist_\* ( \bolda, \, \boldb)$, where  $\*  \in \{\infty,S,H\}$, to denote  either  one of these distances, as appropriate.  Note that
\beq
\dist_{H} ( \bolda, \, \boldb)\le \dist_{S} ( \bolda, \, \boldb) \le \dist_{\infty} ( \bolda, \, \boldb).
\eeq
If $n=2$, we have $\dist_{H} ( \bolda, \, \boldb)= \dist_{S} ( \bolda, \, \boldb)$ for all $\bolda, \boldb  \in \R^{2d}$.

\item The $n$-particle $\*$-box of side $L\geq 1$, where  $\* \in \{\infty,S,H\}$,  centered at $\x\in\R^{nd}$, is defined by
\begin{align} \label{*boxes}
\blam^{(n)}_{\*; L} (\x) &= \set{\y \in \Z^{nd};\,  \dist_{\*} \pa{\x, \y} \le \tfrac{L}{2}}.
 \end{align}
    $\blam^{(n)}_{\infty; L} (\x)=\blam^{(n)}_{ L} (\x)$ is the $n$-particle box of side $L$ centered at $\x$, $\blam^{(n)}_{\paS; L} (\x)$ is the symmetrized  $n$-particle box of side $L$, and $\blam^{(n)}_{H; L} (\x)$ is the Hausdorff  $n$-particle box of side $L$.  Note that
 \begin{align} \label{boxid}
 \blam^{(n)}_{\paS; L} (\x)    =      \bigcup_{\pi \in \cP_n} \blam_{L} (\pi(\x) ) \subset \blam^{(n)}_{H; L} (\x)    \subset    \bigcup_{\y\in \cS_{\x}^n} \blam_{L} (\y ).
 \end{align}
 We also define $n$-particle $\*$-rectangles for  $\* \in \{\infty,S\}$, centered at   $\x\in\R^{nd}$ with sides $\bL=(L_1,L_2,\ldots,L_n)\in [1,\infty)^n$,  by
 \begin{align}\label{rectangle}
  {\blam}_{\bL}^{(n)} (\x)&= \prod_{j=1}^n \La_{L_j}(x_j),  \\ \notag
  {\blam}_{\paS; \bL}^{(n)} (\x)&= \bigcup_{\pi \in \cP_n} \pi\pa{ {\blam}_{\bL}^{(n)} (\x)}
  = \bigcup_{\pi \in \cP_n}  \prod_{j=1}^n \La_{L_{\pi(j)}}(x_{\pi(j)}) .
  \end{align}

\item  We denote by $\{\delta_\x;\,\x\in\Z^{nd}\}$ the canonical orthonormal base of $\ell^2(\Z^{nd})$.

\item  Given $\blam_1 \subset \blam_2 \subset \Z^{nd}$, we set
\[\partial^{{\blam}_2}{{\blam}_1} =\{(\boldu,\boldv)\in \blam_1\times(\blam_2\setminus\blam_1)\,; \norm{\boldu-\boldv}_1=1 \}, \]
where $\norm{\cdot}_1$ is the graph-norm, (i.e., the $1$-norm) in $\Z^{nd}$,
\[\partial_+^{{\blam}_2}{{\blam}_1}=\{\boldv\in {\blam}_2 \setminus {\blam}_1\,; (\boldu,\boldv)\in\partial^{{\blam}_2}{{\blam}_1} \,\text{for some}\,\boldu\in\blam_1 \}, \]
\[ \partial_-^{{\blam}_2}{{\blam}_1}=\{\boldu\in \blam_1\,; (\boldu,\boldv)\in\partial^{{\blam}_2}{{\blam}_1} \,\text{for some}\,\boldv\in{\blam}_2 \setminus {\blam}_1\}.\]
If $\blam_2=\Z^{nd}$, it may be ommitted from the notation.

 \item
 Given $\bTheta\subset \Z^{nd} $, we define the finite-volume operator $H_{\bTheta}^{(n)}$ as the restriction of $ \chi_{\bTheta}H^{(n)}\chi_{\bTheta}$ to $\ell^2(\bTheta)$.  If $z\notin \sigma(H_{\bTheta}^{(n)})$, we set $G_{\bTheta}(z)=(H_{\bTheta}^{(n)}-z)^{-1}$ and $G_{\bTheta}(z,\boldu,\y)=\angles{\delta_\boldu,(H_{\bTheta}^{(n)}-z)^{-1}\delta_\y}$ for $\boldu,\boldy \in \bTheta$.

\item  Given an open interval $I \subseteq \R$, we denote  $\Ccinf (I)$ to be the class of real-valued infinitely differentiable functions on $\R$ with compact support contained in $I$, with $\Ccinfp (I)$ being the subclass of nonnegative functions.

\end{enumerate}

The  random $\*$-moment of order $p \ge 0$ at time $t$ for the time evolution, initially  spatially  localized in the  $n$-particle box of side one around the point $\y \in \Z^{nd}$,  and localized in energy by a function $g \in \Ccinfp (\R)$,  is given by
 \beq \label{randommo}
M_{\bom}^{(n,\*)} \pa{ p, g, t, \y} = \sum_{\boldu \in\Z^{nd}} \scal{\dist_\*(\y, \boldu)}^p \abs{ \scal{\delta_{\boldu}, e^{-itH_{\bom}^{(n)}} g\pa{H_{\bom}^{(n)}}  \delta_{\y} } }^2.
\eeq
The expectation of the random $\*$-moment is given by
\beq \label{exprandommo}
\textbf{M}^{(n,\*)} \pa{ p, g, t, \y} = \E \pa{M_{\bom}^{(n,\*)} \pa{ p, g, t, \y} },
\eeq
  and its time-averaged expectation is defined as
  \beq \label{timeavg}
  \cM^{(n,\*)} \pa{ p, g, T,  \y  } = \frac{2}{T} \int_{0}^{\infty} e^{-\frac{2t}{T}} \textbf{M}^{(n,\*)} \pa{ p, g, t, \y} dt,
  \eeq

 \begin{remark}Note that both ${\normalfont\textbf{M}}^{(n,\*)} \pa{ p, g, t, \y}$ and $ \cM^{(n,\*)} \pa{ p, g,  T, \y  } $ are invariant under the action of $\Z^d$, that is, for all $a\in \Z^d$ we have
  ${\normalfont\textbf{M}}^{(n,\*)} \pa{ p, g, t, \y}={\normalfont\textbf{M}}^{(n,\*)} \pa{ p, g, t, \tau_a(\y)}$  and  $\cM^{(n,\*)} \pa{ p, g, T, \y}=\cM^{(n,\*)} \pa{ p, g,T, \tau_a(\y)}$,  where  $\tau_a(\y)=(y_1-a,..,y_n-a)$.
 \end{remark}

For $p\geq 0$ and non-zero $g\in C_{c,+}^{\infty}(\R)$, we consider the upper and lower transport exponents, defined by
\begin{align}
\label{beta2pm} \beta_{(n,\*)}^+(p,g)&:=\limsup_{T\rightarrow \infty} \frac{\log \displaystyle\sup_{\y\in\Z^{nd}} \cM^{(n,\*)}(p,g,T,\boldy)}{p\log T}, \\ \notag
\beta_{(n,\*)}^-(p,g)&:=\liminf_{T\rightarrow \infty} \frac{\log\displaystyle\sup_{\y\in\Z^{nd}} \cM^{(n,\*)}(p,g,T,\boldy)}{p\log T}.
\end{align}

Note that these quantities are well defined.   Ergodicity implies that either  $g(H_{\bom}^{(n)})\neq 0$ almost surely or  $g(H_{\bom}^{(n)})=0$ almost surely.  If $g(H_{\bom}^{(n)})\neq 0$ almost surely, we have  $\sup_{\y\in\Z^{nd}} \cM^{(n,\*)}(p,g,T,\boldy)>0$ (see Remark \ref{posi}). If $g(H_{\bom}^{(n)})=0$ almost surely, we set  $\beta_{(n,\*)}^\pm(p,g)=0$.

 Following \cite[Eq. 2.16 - 2.20]{GKduke}, we define the $p$-th upper and lower transport exponents in an open interval $I$, and the  $p$-th local upper and lower transport exponents at an energy $E\in\R$,  by
\begin{align} \beta_{(n,\*)}^\pm(p,I)&=\sup_{g\in C_{c,+}^{\infty}(I)}\beta_{(n,\*)}^\pm(p,g),\\ \notag
 \beta_{(n,\*)}^\pm(p,E)& =\inf_{I\ni E, \ I \text{ open interval}} \beta_{(n,\*)}^\pm(p,I).
 \end{align}
The asymptotic upper and lower transport exponents and the local asymptotic transport exponents are defined in the same way as in \cite[Eq. 2.16-2.20]{GKduke}, using the fact that the exponents defined above are non-decreasing in $p$ (see Proposition \ref{propbeta}),
\begin{align} \beta_{(n,\*)}^\pm(I)&=\lim_{p\rightarrow\infty}\beta_{(n,\*)}^\pm(p,I)=\sup_p\beta_{(n,\*)}^\pm(p,I), \\
\notag \beta_{(n,\*)}^\pm(E)&= \lim_{p\rightarrow\infty}\beta_{(n,\*)}^\pm(p,E)=\sup_p\beta_{(n,\*)}^\pm(p,I).
\end{align}

In the following, we list  properties of the (random) moments and their consequences for the transport exponents. { These statements are proven similarly to their continuous analogs, proven in \cite{GKduke}.}

\begin{proposition}\label{prop3.1GKduke} Let $g\in C_{c,+}^\infty(\R)$, $g(H_\om^{(n)})\neq 0$ with probability one. We have, for every $\boldy\in\Z^{nd}$ and $p\geq 0$,
\begin{align}
 0 & \leq M_{\om}^{(n,\*)}(0,g,0,\y)\leq M_{\om}^{(n,\*)}(p,g,t,\y)\leq C_{d,g,p}\angles{t}^{\lfloor p+nd \rfloor+2} \label{m1}\\
 0 & \leq \bold M^{(n,\*)}(0,g,0,\y)\leq \bold M^{(n,\*)}(p,g,t,\y)\leq C_{d,g,p}\angles{t}^{\lfloor p+nd \rfloor+2}\label{m2}\\
 0 & \leq \bold M^{(n,\*)}(0,g,0,\y)\leq \cM^{(n,\*)}(p,g,T,\y)\leq C_{d,g,p}\angles{T}^{\lfloor p+nd \rfloor+2},\label{m3}
\end{align}
where $\lfloor x \rfloor$ denotes the largest integer less than or equal to $x\in\R$.
\end{proposition}

We have \cite[Proposition 3.2]{GKduke},
\begin{proposition}\label{propbeta}Let $g\in C_{c,+}^\infty(\R)$  and $E\in\R$, then
\begin{itemize}
\item[(i)]  $\beta_{(n,\*)}^\pm(p,g)$ is monotone increasing in $p\geq 0$.
\item[(ii)] $0\leq \beta_{(n,\*)}^\pm(p,g)\leq 1$.
\end{itemize}
\end{proposition}

\begin{remark}\label{posi}For simplicity, we write $g(H^{(n)}_\om)=g(H)$. Let us look at the quantity
\[ \E\pa{\norm{g(H)\delta_\boldu}^2}={\bf M}^{(n,\*)}(0,g,0,\boldu)\]
Since $\set{\delta_\boldu}_{\boldu\in\Z^{nd}}$ is an orthonormal base, we have that ${\bf M}^{(n,\*)}(0,g,0,\boldu)=0$ for all $\boldu\in\Z^{nd}$ implies $g(H)=0$ almost surely.
Let us suppose $g(H)\neq 0$. Then, there exists at least one $\boldu\in\Z^{nd}$ such that ${\bf M}^{(n,\*)}(0,g,0,\boldu)>0$, and by \eqref{m3}, $\cM^{(n,\*)}(p,g,T,u)>0$. This implies
\be \sup_{\y\in\Z^{nd}} \cM^{(n,\*)}(p,g,T,\y)>0. \ee
\end{remark}

\bigskip
\begin{definition} We say $H_{\bom}^{(n)}$ exhibits  strong dynamical localization \emph{in the $\*$-distance} in an open interval $I$ if for all $g\in C_{c,+}^\infty(I)$ we have
\be \sup_{\y\in\Z^{nd}}\E\lbrace \sup_{t\in\R} M_\omega^{(n,\*)}(p,g,t,\y) \rbrace<\infty \ee
for all $p\geq 0$. We say $H_{\bom}^{(n)}$ exhibits strong dynamical localization  in the $\*$-distance   at an energy $E\in\R$ if there exists an open interval $I$ containing $E$ such that $H_{\bom}^{(n)}$ exhibits  strong dynamical localization  in the $\*$-distance  in $I$.

We denote by $\Sigma^{(n,\*)}_{\rm SI}$ the corresponding region of strong insulation, that is, the set of energies for which $H_{\bom}^{(n)}$  exhibits strong dynamical localization in the $\*$-distance .
\end{definition}

\begin{definition}\label{inputMSA} Let  $\blam=\blam_{\*;L}^{(n)}(\y)$ be the $n$-particle $\*$-box of center $\y\in \R^{nd}$ and side-length $L$.
\begin{itemize}
\item[i)] Given $\theta>0$, $E\in \R$, we say that  $\blam$ is $(\theta,E)$-\emph{suitable} if $E\notin\sigma(H_{\blam}^{(n)})$ and
    \be \abs{G_{\blam}(\boldu,\boldv;E)}\leq L^{-\theta}
    \mqtx{for all}  \boldu \in \blam_{\*;L/3}(\y) \qtx{and} \boldv\in\partial_-\blam.
    \ee
  Otherwise, $\blam$ is $(\theta,E)$-\emph{nonsuitable}.
\item[ii)] Given $m>0$, $E\in \R$, we say that  $\blam$ is $(m,E)$-\emph{regular} if $E\notin\sigma(H_{\blam}^{(n)})$ and
    \be \abs{G_{\blam}(\boldu,\boldv;E)}\leq e^{-mL/2}  \mqtx{for all}  \boldu \in \blam_{\*;L/3}(\y) \mqtx{and} \boldv\in\partial_-\blam.
    \ee
 Otherwise, $\blam$ is  $(m,E)$-\emph{nonregular}.

 \item[iii)]  Given  $\zeta \in (0,1)$, $E\in \R$, we say  that  $\blam$ is  $(\zeta, E)$-subexponentially suitable (SES) if, and only if,
$E \notin \sigma \Bl(H_{\boldlambda}  \Br)$  and
\beq
 \abs{G_{{\blam}}(\boldu, \boldv; E)}
\leq e^{-L^{\zeta}} \mqtx{for all}  \boldu \in \blam_{\*;L/3}(\y) \mqtx{and} \boldv\in\partial_-\blam.
\eeq
 Otherwise,  ${\blam}$ is called  $(\zeta, E)$-nonsubexponentially suitable (nonSES).

    \end{itemize}
\end{definition}

 \begin{theorem} \label{mainthmdisc}
 Fix  $\*  \in \{\infty,S,H\}$, $n\in\N$, let $H_{\bom}^{(n)}$ be  the $n$-particle Anderson model and let $\I_n \subset\R$ be an open interval. Let $g_n \in \Ccinfp(\R)$ such that $g_n\equiv 1$ on $\I_{n}$.  If, for some   $\alpha\ge 0$, $\theta>2nd$, and  $p > p(\alpha,n,d,\theta)=(\theta+2nd)\alpha+6\theta+15nd $, we have
   \begin{align}\label{conddldisc}
 &\liminf_{T  \to  \infty}    \frac{1}{T^{\alpha}}   \sup_{\y \in \Z^{nd}}   \cM^{(n,\*)} \pa{ p, g_n,  T, \y  } <  \infty,
  \end{align}
 then there exists a sequence of increasing length-scales $L_k$ such that
  \be\label{ILSEk} \lim_{k\rightarrow\infty}\sup_{\y\in\R^{nd}}\P\lbrace\blam_{\*;L_k}^{(n)}(\y) \,\mbox{is }(\theta,E)\text{-nonsuitable}\rbrace =0 \mqtx{for all} E\in \mathcal I_n.
  \ee
 \end{theorem}

We define  the energy  region of trivial transport $\Sigma^{(n,\*)}_{\rm TT}$ by
\be \Sigma^{(n,\*)}_{\rm TT}=\lbrace E\in \R,\, \beta_{(n,\*)}^-(E)=0\rbrace.\ee
Note that  we have
\beq
\R \setminus \Sigma^{(n)}\subset   \Sigma^{(n,\*)}_{\rm SI}\subset \Sigma^{(n,\*)}_{\rm TT}.
\eeq

 If $n \ge 2$, the Multiscale Analysis for the $n$-particle   Anderson model can only be performed in an interval at the bottom of the spectrum, and requires a Wegner estimate between boxes. This estimate is known for $n$-particles  boxes separated in the Hausdorff distance (e.g., \cite[Corollary 2.4]{KlN1}), which allows the performance of the multiscale analysis for  the $n$-particle   Anderson model based on $\infty$-boxes, yielding dynamical localization in the Hausdorff distance (e.g., the discrete analogue of  \cite[Theorem~1.6]{KlN2}).

The conclusions of Theorem~\ref{mainthmdisc}  are statements about $\*$-boxes.  \cite[Theorem~1.6]{KlN2} requires an initial step which is a statement about $\infty$-boxes, yielding dynamical localization in the Hausdorff distance. If we apply Theorem~\ref{mainthmdisc} to the conclusions of \cite[Theorem~1.6]{KlN2},  we would obtain a statement about $H$-boxes.
To go back we would need to perform a multiscale analysis using $H$-boxes.

There are technical problems with performing a multiscale analysis based on $H$-boxes. But these problems are not present for a multiscale analysis based on $S$-boxes. On the other hand, there is no Wegner estimate between boxes separated in the symmetrized distance, except for the case of two particles, where the symmetrized and the Hausdorff distance coincide.  For this reason we will now restrict ourselves to the $2$-particle Anderson model, for which we prove the following theorem, the analog of \cite[Theorem 1.6]{KlN2} for symmetrized two-particle boxes (see also \cite{GK1}).

\begin{theorem}[Bootstrap Multiscale Analysis for  two-particle $S$-boxes]\label{MSAthm}    Let $H_{\bom}^{(2)}$ be the $2$-particle Anderson model.
There exist $p_0(n)=p_0(n,d)>0$, n=1,2, such that, given $\theta>16d$ and energies $E^{(1)}>E^{(2)}>0$, there exists $\cL=\cL(d,\norm{\rho}_\infty,\theta,E^{(1)},E^{(2)})$ such that if for some $L_0\geq \cL$ and $n=1,2$ we have
\begin{equation} \label{hypMSA}
\sup_{x\in\R^{nd}} \P\set{\blam^{(n)}_{S;L_0}(\x)\sqtx{is}(\theta,E)\text{-nonsuitable}}\leq p_0(n) \quad\mbox{for all }E\leq E^{(n)},
\end{equation}
then, given $0<\zeta<1$, we can find a length scale $L_\zeta=L_\zeta(d,\norm{\rho}_\infty,\theta,E^{(1)},E^{(2)},L_0)$,
$\delta_\zeta=\delta_\zeta(d,\norm{\rho}_\infty,\theta,E^{(1)},E^{(2)},L_0)>0$ and $m_\zeta=m_\zeta(L_\zeta,\delta_\zeta)>0$, such that, for $n=1,2$, we have that
for every $E_1<E^{(n)}$, $L\geq L_\zeta$ and all  $\bolda,\boldb\in\R^{nd}$ with $d_S(\bolda,\boldb)>L$, we have
 \begin{equation}
\P\set{\blam^{(n)}_{S;L}(\bolda)\sqtx{and}\blam^{(n)}_{S;L}(\boldb)\sqtx{are}(m_\zeta,E)\text{-nonregular} \sqtx{for some}E\in I(E_1)}\leq e^{-L^\zeta},
\end{equation}
where $I(E_1)=[E_1-\delta_\zeta,E_1+\delta_\zeta]\cap (-\infty,E^{(n)}]$.
\end{theorem}

 Unlike the one particle case, this multiscale analysis can only be performed in an interval at the bottom of the spectrum, and it requires also the performance of the multiscale analysis for one-particle in a slightly larger interval.  For this reason $\Sigma^{(2,S)}_{\rm MSA}$, the set of energies   where we can start the multiscale analysis for $2$-particle  $S$-boxes, has to be  defined differently from  \cite[Definition~2.6]{GKduke}.    We say $  E^{(2)} \in  \cE^{(2,S)}_{\rm MSA}$ if $E^{(2)} >0$ and there exists $E^{(1)}>E^{(2)}$ such that the condition \eq{hypMSA} holds for $n=1,2$
for some $L_0\geq \cL$, where $\cL=\cL(d,\norm{\rho}_\infty,\theta,E^{(1)},E^{(2)})$ is as in   Theorem \ref{MSAthm}.  We set
\beq\label{newdefSigmaMSA}
\Sigma^{(2,S)}_{\rm MSA}=\pa{-\infty,  E^{(2,S)}_{\rm MSA}},\qtx{where} E^{(2,S)}_{\rm MSA}=\sup \cE^{(2,S)}_{\rm MSA}.
\eeq

For $n=1$, Theorem \ref{mainthmdisc} corresponds to \cite[Theorem 2.11]{GKduke}. Proceeding as in
\cite[Theorem 3.21]{KlN1}, we can extend the estimate \eqref{ILSEk}, which holds for
 a sequence of length-scales $L_k$, to a corresponding estimate that holds for every large enough scale $L$ (see Remark \ref{lktocont}).  This yields the following Corollary.

\begin{corollary}\label{cormainthm} Assume the  hypotheses of Theorem \ref{mainthmdisc} hold  for $n=1,2$ (with $\* =H=S$) on an interval  $(-\infty,E_*)$, where    $E_*>0$. Then   $(-\infty,E_*) \subset \Sigma^{(2,S)}_{\rm MSA}$.
\end{corollary}

We define:
\beq
\widetilde\Sigma^{(2,\*)}_{\rm TT} = (-\infty, E^{(2,\*)}_{\rm TT})\qtx{and} \widetilde\Sigma^{(2,\*)}_{\rm SI} = (-\infty, E^{(2,\*)}_{\rm SI}),
\eeq
where
\begin{align}
E^{(2,\*)}_{\rm TT} &=\sup \set{E; (-\infty,E) \subset \Sigma^{1}_{\rm TT} \cap\Sigma^{(2,\*)}_{\rm TT}  },\\ \notag
E^{(2,\*)}_{\rm SI} &=\sup \set{E; (-\infty,E) \subset \Sigma^{1}_{\rm SI} \cap\Sigma^{(2,\*)}_{\rm SI}  }
\end{align}

\begin{theorem}\label{dynloc}
Let $H_{\bom}^{(2)}$ be the $2$-particle Anderson model. Then,
$\Sigma^{(2,S)}_{\rm MSA}\subset\widetilde\Sigma^{(2,S)}_{\rm SI} $. In particular, $\Sigma^{(2,S)}_{\rm MSA}\subset\widetilde \Sigma^{(2,S)}_{\rm TT}$.
\end{theorem}

In the following Theorem, we show that for the two-particle Anderson model, the converse to Theorem~\ref{dynloc} holds true. This gives a characterization of the metal-insulator transition for the two-particle Anderson model at the bottom of the spectrum,  analogous to \cite[Theorem 2.8]{GKduke}.

\begin{theorem}\label{thm2.8and2.9GKduke} Let $H_{\bom}^{(2)}$ be the $2$-particle Anderson model. Then,
$\widetilde\Sigma^{(2,\*)}_{\rm TT} \subset\Sigma^{(2,S)}_{\rm MSA}$, which yields
\be \Sigma^{(2,S)}_{\rm MSA}= \widetilde\Sigma^{(2,S)}_{\rm SI}=\widetilde\Sigma^{(2,\*)}_{\rm TT} .\ee
\end{theorem}

 Theorem \ref{thm2.8and2.9GKduke} is a consequence of Theorems \ref{mainthmdisc} and \ref{dynloc}.  It shows that for the $2$-particle Anderson model,  within the region of one-particle dynamical localization, slow transport (i.e.,  $\beta_{(2,S)}^-(p,E)$ small)  at the bottom of the spectrum implies dynamical localization.  In particular, it implies null transport (i.e., $\beta_{(2,S)}^-(p,E)=0$)  in an interval at the bottom of the spectrum.

\begin{remark}  We can obtain the following analogue of  \cite[Theorem 2.10]{GKduke}.  Note that we have $E^{(2,S)}_{\rm MSA}\le E^{(1)}_{\rm MSA}$ (recall \eq{newdefSigmaMSA}). As in the proof of \cite[Theorem 2.10]{GKduke}, it follows from Theorem \ref{mainthmdisc} applied to the one-particle setting that $\beta_{1}^-(E^{(1)}_{\rm MSA})\ge \frac 1 {3d}$.
 If $E^{(2,S)}_{\rm MSA}=E^{(1)}_{\rm MSA}$ this is all we can say.  If  $E^{(2,S)}_{\rm MSA}< E^{(1)}_{\rm MSA}$,
it follows from Theorem ~\ref{mainthmdisc} that we must have $\beta_{(2,S)}^-(E^{(2,S)}_{\rm MSA})
\geq \frac{1}{20d}$.
\end{remark}

\begin{remark} Recently Chulaevsky obtained a Wegner estimate between boxes that are far apart in the symmetrized distance for all $N$, for a certain class of single site probability distributions \cite[Theorem 2.2]{Ch2} (see also \cite{Ch1}).  Under Chulaevsky's assumptions we would have Corollary \ref{cormainthm} and Theorem \ref{dynloc}, and therefore Theorem \ref{thm2.8and2.9GKduke}, for any $N\geq2$.
 \end{remark}

\bla
\section{Proof of Theorem \ref{mainthmdisc}}\label{S:proof1}

 We prove Theorem \ref{mainthmdisc} following  \cite{GKduke}.  We need to adapt
\cite[Proposition 6.1 and Lemma 6.4]{GKduke} to the discrete setting and distance $\dist_\*$, $\*\in\{\infty,S,H\}$.

We start with \cite[Proposition 6.1]{GKduke}.

\begin{theorem}\label{prop6.1GKduke}
 Let $H_{\bom}^{(n)}$ be a random $n$-particle Schr\"odinger operator, $g\in C_{c,+}^\infty(\R)$. For any $\boldu\in\Z^{nd}$, $p>0$ and $T>0$, we have
 \begin{align}
 & \cM^{(n,\*)} (p,g,T,\boldu) \\
 & \hskip 18pt =\frac{1}{\pi T} \int_\R   \sum_{\boldv \in \Z^{nd}}\scal{ \dist_{{\*}}(\boldv,\boldu)}^{p} \E \pa{ \abs{\scal{\delta_{\boldv}, G_{\bom}\left(E+\tfrac{i}{T}  \right) g\pa{H_{\bom}^{(n)}} \delta_{\boldu}} }^2 } dE. \notag
 \end{align}
\end{theorem}

The proof of Theorem \ref{prop6.1GKduke} is similar to the proof of \cite[Proposition 6.1]{GKduke}, with the corresponding changes to the discrete setting, and considering the multiplication operator $\dist_\*(\boldu,\cdot)$ instead of the usual position operator $\angles{\cdot}=\dist_\infty(\boldu,\cdot)$, so we omit the proof.

The following is a generalization of \cite[Lemma 6.4]{GKduke} to the discrete setting, we give a proof in Appendix~\ref{appdiscrete} (see Lemma~\ref{aGKlem6.4discA}) for the reader's convenience.

\begin{lemma}\label{GKlem6.4disc}
Let $H_{\bom}^{(n)}$ be a random $n$-particle Schr\"odinger operator satisfying a Wegner estimate for $\*$-boxes of the form \eqref{wesharp} in an open interval $\mathcal I$. Let us denote by $\blam$ the  $n$-particle $\*$-box $\blam_{\*;L}^{(n)}$, where $\*\in\{\infty,S,H\}$. Let $p_0>0$ and $\gamma>nd$.
 There exists a scale ${\mathcal L}={\mathcal L}(\gamma,n,d,\rho,p_0)$ such that, given  $E\in \mathcal I$, $L\geq \mathcal L$, and subsets $B_1, B_2\subset \blam$ such that $\partial_-\blam\subset B_2$,  for each $a>0$ and $\eps>0$ we have, for $\boldu\in B_1$ and $\y\in B_2$
\begin{align} \label{GKlem6.4adisc}
& \P \pa{ a< \abs{G_{{\blam}}(E+i\eps;\boldy,\boldu)} } \\
& \hskip 40pt \leq \frac{4L^{\gamma+2nd}}{a} \sup_{\boldk \in\partial_{+} {\blam} \cup B_2}\E\left( \abs{G(E+i\eps;\boldk,\boldu)}\right)+ p_0,  \notag
\end{align}
and
\begin{align}\label{GKlem6.4bdisc}
& \P \pa{ a< \abs{ G_{{\blam}}(E;\boldy,\boldu)} }  \\
& \hskip 20pt \leq  \frac{ 8L^{\gamma+2nd}}{a}  \sup_{\boldk\in\partial_{+} {\blam}\cup B_2}\E\left( \abs{G(E+i\eps;\boldk,\boldu)}\right)+ 2^{3/2}C_n^{(\sharp)}\norm{\rho}_\infty\sqrt{\frac{\eps}{a}}L^{nd} +p_0. \notag
\end{align}
\end{lemma}

We are now ready to prove Theorem \ref{mainthmdisc}.

\begin{proof}[Proof of Theorem \ref{mainthmdisc}]
Fix $n\geq 2$.
 We write for simplicity ${\blam} = {\blam}_{\*; L}^{(n)}(\boldx) $.  We need to show that for $\theta >0$ and any $\tilde p_0:=\tilde p_0(n)>0$ there exists $\cL_0$ such that, given $\cL_1 \ge \cL_0$, for some  $L\ge \cL_1$ we have \be\label{ilse}
P_{E,L}:=\P \set{\abs{ G_{{\blam}}(E;\boldy,\boldu)}   >  L^{-\theta}} < \tilde p_0,
\ee
for all $\boldu\in{\blam}_{\*; {L/3}}^{(n)}(\x)$ and $\y\in\partial_{-}{\blam}$ and all $E\in\I_n$, uniformly in $\x\in\R^{nd}$. We will proceed as in the proof of \cite[Theorem 2.11]{GKduke}. Let $p_0>0$ and use Lemma \ref{GKlem6.4disc} with $a=8 L^{-\theta}$, $B_1=\blam_{\*;L/3},B_2=\partial_-\blam$. This gives the existence of ${\mathcal L}(\gamma,n,d,\rho,p_0)$ such that for $L\geq {\mathcal L}(\gamma,n,d,\rho,p_0)$, we have

\be P_{E,L}\leq  L^{\gamma+2nd+\theta}  \sup_{\boldk\in\partial_\pm {\blam} }\E\left( \abs{G(E+i\eps;\boldk,\boldu)}\right)+ C_{\rho,n}\sqrt{\eps}L^{nd+\theta/2} +p_0, \ee
where $ C_{\rho,n}:=C_n^{(\sharp)}\norm{\rho}_\infty $ and $\partial_{\pm} \blam:=\partial_+\blam\cup\partial_-\blam$ . We will make the second term in the r.h.s. small by taking
\be\label{Leps} L=L(\eps)=\left(\frac{p_0}{2C_{\rho,n}\sqrt{\eps}}\right)^{\frac{2}{\theta+2nd}}.\ee
Then,
\be\label{prob3} P_{E,L}\leq  L^{\gamma+2nd+\theta}  \sup_{\boldk\in\partial_\pm{\blam} }\E\left( \abs{G(E+i\eps;\boldk,\boldu)}\right)+ 2p_0. \ee
We will split the first term in the r.h.s. of \eqref{prob3} as in \cite[Eq. 6.31]{GKduke} using the fact that $g_n(H_\omega^{(n)})\equiv 1$ on $\mathcal I_n$. We will drop the subscript from $g_n$ for simplicity. Next, we use \cite[Theorem A.5]{GKduke}, which holds in the discrete setting, see \cite[Remark 2.1]{GKduke}, to obtain
\be
 L^{\gamma+2nd+\theta} \sup_{\boldk\in\partial_\pm {\blam}}\E\left( \abs{\scal{\delta_\boldk, G(E+i\eps)\left(1-g(H_\omega^{(n)})\right)\delta_{\boldu}}}\right)<p_0,
 \ee
 where we used that $\dist_\infty(\boldk,\boldu)>\dist_\*(\boldk,\boldu)$ and $\dist_\*(\boldk,\boldu)>L/2$.
 We obtain
\begin{align}\label{prob4}
& P_{E,L}  \leq L^{\gamma+2nd+\theta} \sup_{\boldk\in\partial_\pm {\blam} }\E\left( \abs{\scal{\delta_\boldk, G(E+i\eps)g(H_\omega^{(n)})\delta_{\boldu}}}\right) +3p_0\\ \nonumber
& \quad \leq C_pL^{-p/2+\gamma+2nd+\theta}  \sup_{\boldk\in\partial_\pm {\blam} }\E\left( \scal{\dist_\*(\boldk,\boldu)}^{p/2}\abs{\scal{\delta_\boldk, G(E+i\eps)g(H_\omega^{(n)})\delta_{\boldu}}}\right)\\  \notag   &   \qquad  \qquad  \qquad  +3p_0,
 \end{align}
where $ C_{p}:=6^{p/2}$ and we used the fact that $\boldk\in\partial_\pm{\blam} $ implies $\dist_\*(\boldu,\boldk)>L/6$ for $L>1$.
We use Jensen's inequality to obtain
\begin{align} & \sup_{\boldk\in\partial_\pm {\blam} }\E\left( \scal{\dist_\*(\boldk,\boldu)}^{p/2}\abs{\scal{\delta_\boldk, G(E+i\eps)g(H_\omega^{(n)})\delta_{\boldu}}}\right)\\
&\quad\quad\quad \leq \E\left(\sum_{ \boldv \in \Z^{nd}} \scal{\dist_\*( \boldv ,\boldu)}^{p} \abs{\scal{ \delta_{\boldv}, G(E+i\eps) g(H_\omega^{(n)})\delta_{\boldu}  }}^2 \right)^{1/2} .\nonumber\end{align}

 Next, we define
 \begin{align}
& A_{\boldu,M,I,\eps} := \\
& \hskip 2pt \set{ E  \in I \,:\,  \E   \pa{  \sum_{\boldv\in\Z^{nd}} \scal{\dist_\* (\boldv,\boldu)}^{p} \abs{\scal{\delta_{\boldv},G(E+i\eps)g(H_\omega^{(n)})\delta_{\boldu} }}^2  }\leq M \eps^{-(\alpha+1)}   } .\notag
\end{align}
Taking $T=\eps^{-1}$ and using Theorem \ref{prop6.1GKduke}, we get
\beq
\abs{I\setminus A_{\boldu,M,I,\eps}} \leq \frac{\pi}{MT^{\alpha}}\sup_{\boldu\in\Z^{nd}} \cM^{(n,\*)}(p,T,g,\boldu).
\eeq
By our hypothesis \eqref{conddldisc}, we can pick a sequence $T_j\rightarrow\infty$ such that, for $j$ big enough, we have $\displaystyle\sup_{\boldu\in\Z^{nd}} \mathcal M^{(n,\*)}(p,g,T_j,\boldu) <C T_j^\alpha$ for some
positive constant $C$. Then, for the corresponding sequence $\eps_j:=T_j^{-1} \rightarrow 0^+$  we have
\be\label{Acomp1disc} |I\setminus A_{\boldu,M, I,\eps_j}| \leq \frac{C'}{M}, \ee
where $C'=\pi C$. Note that this  bound is uniform in $\boldu$.

For an $E\in I$ fixed, either $E\in A_{\boldu, I,M,\eps_j}$ or $E\in I\setminus A_{\boldu,M,I,\eps_j}$. In the first case we have,

\be\label{probEu}
P_{E,L_j} \leq C_{p}L_j^{-p/2+\gamma+2nd+\theta} M^{1/2}\eps_j^{-(\alpha+1)/2}+3p_0
 \ee
where we write $L_j:=L(\eps_j)$. If $E\in I\setminus A_{\boldu, M,I,\eps_j} $, by \eqref{Acomp1disc}  there exists $E_{\boldu}\in A_{\boldu, I,M,\eps_j}$ such that
$ |E-E_{\boldu}|\leq \frac{ C'}{ M}$.
Using the resolvent identity, we obtain
\beq
\E \pa{ \abs{ \angles{\delta_\boldk, G(E+i\eps_j)g(H_\omega^{(n)})\delta_\boldu}}  }\leq \E \pa{ \abs{ \angles{\delta_\boldk, G(E_\boldu+i\eps_j)g(H_\omega^{(n)})\delta_\boldu}}  } + \frac{C'}{M\eps_j^2}.
\eeq
 It follows that
\begin{align}
P_{E,L_j}& \leq L_j^{\gamma+2nd+\theta} \sup_{\boldk\in\partial_\pm {\blam} }\E\left( \abs{\scal{\delta_\boldk, G(E_\boldu+i\eps_j)g(H_\omega^{(n)})\delta_{\boldu}}}\right) \nonumber \\
& \quad \quad \quad +\frac{C'L_j^{\gamma+2nd+\theta}}{M\eps_j^2} +3p_0.
 \end{align}
 We can bound the first term in the r.h.s. as in \eqref{prob4}-\eqref{probEu} and get
\be\label{prob5}
P_{E,L_j}\leq  C_{p}L_j^{-p/2+\gamma+2nd+\theta} M^{1/2}\eps_j^{-(\alpha+1)/2} +\frac{C'L_j^{\gamma+2nd+\theta}}{M\eps_j^2} +3p_0.
 \ee
We set $M= L_j^{7\gamma+3\theta}$. Recalling \eqref{Leps} and $\gamma>nd$, we can take $p$ such that $p>p(\alpha,n,d,\theta)=(\theta+2nd)\alpha +6\theta+15nd$,  and find $\gamma$ and $L_j$ larger than some scale $\cL=\cL(d,p,\alpha,\theta,\gamma,p_0,C_{\rho,n},C')$, such that the r.h.s. of \eqref{prob5} is bounded by $5p_0$.
Therefore, there exists a sequence $L_j\rightarrow \infty$ such that for $L_j$ large enough and $p_0<\tilde p_0/5$ we have $P_{E,L_j} <\tilde p_0$ uniformly on $\x\in\R^{nd}$. Since $\tilde p_0$ is arbitrary, we obtain the desired result.
\end{proof}

\section{The Bootstrap Multiscale Analysis for symmetrized two-particle boxes}\label{S:MSA}
\subsection{Preliminaries}\label{S:preMSA}

Let ${\blam}^{(2)}_{\paS; \bL} (\x)$ be the symmetrized two-particle rectangle  with sides $\bL=(L_1,L_2)$ and center $\x=(x_1,x_2)\in\R^{2d}$ as in \eq{rectangle}, and note that $\abs{{\blam}^{(2)}_{\paS; \bL} (\x)}\leq 2L^{2d}$, where $L=\max\set{L_1,L_2}$.
 We set
\begin{align} \label{projection}
&\Pi_{j}  {\blam}^{(2)}_{L_1,L_2} (\x) = \Lambda_{L_j}(x_j) \qtx{for} j = 1,2 , \notag \\
& \Pi  {\blam}^{(2)}_{L_1,L_2} (\x) \;\, = \Pi  {\blam}^{(2)}_{\paS; L_1,L_2} (\x)  \; \, = \bigcup_{j = 1, 2} \Lambda_{L_j}(x_j).
\end{align}

 \begin{definition} \label{part-ful-sep}
A pair of symmetrized two-particle rectangles,   ${\blam}^{(2)}_{\paS; \bL} (\x)$ and $ {\blam}_{\paS; \bL^{\pr}}^{(2)} (\y)$, are said to be fully separated if and only if
 \beq
 \Pi {\blam}^{(2)}_{\paS; \bL} (\x)  \; \bigcap \; \Pi  {\blam}_{\paS; \bL^{\pr}}^{(2)} (\y) = \emptyset.
 \eeq
The pair is said to be  partially separated if and only if there exists $j \in \set{1, 2}$ such that
 \begin{align}\notag
& \text{either}  \; \; \Pi_j{\blam}^{(2)}_{\bL} (\x) \; \bigcap \;  \Pi  {\blam}_{\paS; \bL^{\pr}}^{(2)} (\y) = \emptyset,  \mqtx{or}
 \Pi_j  {\blam}_{ \bL^{\pr}}^{(2)} (\y)\; \bigcap \; \Pi{\blam}^{(2)}_{\paS; \bL} (\x) = \emptyset.
 \end{align}
 \end{definition}

Note that events defined on a pair of fully separated symmetrized two-particle boxes are independent.

 The following Wegner estimate for partially separated  symmetrized two-particle rectangles can be proven in the same way as in \cite[Theorem 2.3 and Corollary 2.4]{KlN1}, using Theorem~\ref{t:wesharp}.

\begin{proposition}  \label{Wegner2}  Consider a pair of  partially separated  symmetrized two-particle rectangles  ${\blam}^{(2)}_{\paS; \bL} (\x)\subset {\blam}^{(2)}_{\paS; L} (\x)$ and $ {\blam}_{\paS; \bL^{\pr}}^{(2)} (\y)\subset {\blam}^{(2)}_{\paS; L} (\y)$. Then
\beq \label{wegboxes}
\P \Bl\{ d\,\Bl(\sigma (H_{{\blam}^{(2)}_{\paS; \bL} (\x)}), \sigma (H_{{\blam}_{\paS; \bL^{\pr}}^{(2)} (\y)}) \Br) \leq {\eps} \Br\} \leq 16\, \, \norm{\rho}_\infty\, {\eps} \, L^{4d} \mqtx{for all} \eps>0.
\eeq
 \end{proposition}

Let  $\blam_1 \subset \blam_2 \subset \Z^{2d}$.
We have  ${\partial^{{\blam}_2}{{\blam}_1} } \subseteq {\partial{{\blam}_1} } $, so
  if $\blam_1=  {\blam}_{\paS; \ell}^{\pat} (\x)$, $\x \in \R^{2d}$, we have  $\abs{\partial^{{\blam}_2}{{\blam}_1} } \le \abs{\partial{{\blam}_1} } \le  2 s_{2, d} \ell ^{2d-1}$, for some constant $s_{2, d}>0$.

\begin{lemma}Let  ${\blam} = {\blam}_{\paS; \ell}^{\pat} (\x)$, $\x \in \R^{2d}$. If ${(\bolda, \boldb) \in \partial{\blam} }$, we have
\beq
\tfrac \ell 2 -1 < \dist_{{S}} ( \bolda,\x)\le \tfrac \ell 2 <  \dist_{{S}} ( \boldb,\x) \le \tfrac \ell 2 +1.
\eeq
\end{lemma}

\begin{proof} We have
$\bolda \in {\blam}$, so $ \dist_{{S}} ( \bolda,\x)\le \frac \ell 2$,  $\boldb \notin   {\blam}$, so
 $ \dist_{{S}} ( \boldb,\x)>\frac \ell 2\ge  \dist_{{S}} ( \bolda,\x)$, and $ \norm{\bolda - \boldb }_1 = 1$, so $ \norm{\bolda - \boldb } = 1$.

Suppose
$ \dist_{{S}} ( \bolda,\x)= \norm{ \bolda-\x} \le  \norm{ \bolda-\pi(\x)}$. Then
\beq
\norm{ \bolda-\pi(\x)} \ge \norm{ \bolda-\x}\ge \norm{ \boldb-\x}- \norm{ \bolda-\boldb}>\frac \ell 2 -1,
\eeq
and we conclude that $ \dist_{{S}} ( \bolda,\x) > \frac \ell 2 -1$.  Moreover,
\beq
 \dist_{\paS} ( \boldb,\x)\le  \norm{ \bolda-\x} + \norm{ \bolda-\boldb}\le \frac \ell 2 +1.
\eeq
\end{proof}

Let ${\blam}_1 = {\blam}_{\paS; \ell}^{\pat} (\x)$ and $ {\blam}_2 = {\blam}_{\paS; L}^{\pat}(\y)$ with ${\blam}_1 \subset {\blam}_2$, where $\x,\y \in \R^{2d}$.
For $\boldu \in {\blam}_1$,  $\boldv \in {\blam}_2 \setminus {\blam}_1$, and $ z \notin \sigma\pa{H_{{\blam}_1}} \cup \sigma\pa{H_{{\blam}_2}}$, we  have
\begin{align} \label{resolventid}
 \pa{H_{{\blam}_2} - z}^{-1} \pa{\boldu, \boldv}  = \sum_{(\bolda, \boldb) \in \partial^{{\blam}_2}\pa{{\blam}_1} } \pa{H_{{\blam}_1} - z}^{-1} \pa{\boldu, \bolda} \pa{H_{{\blam}_2} - z}^{-1} \pa{\boldb, \boldv}.
\end{align}
Hence, as a consequence of the geometric resolvent identity, we have
\begin{align} \label{resolventine}
& \abs{\pa{H_{{\blam}_2} - z}^{-1} \pa{\boldu, \boldv} }  \\
& \hskip 30 pt \le \abs { \partial^{{\blam}_2}\pa{{\blam}_1}  }  \max_{(\bolda, \boldb) \in \partial^{{\blam}_2}\pa{{\blam}_1} }  \abs{  \pa{H_{{\blam}_1} - z}^{-1} \pa{\boldu, \bolda} \pa{H_{{\blam}_2} - z}^{-1} \pa{\boldb, \boldv}    } \notag \\
& \hskip 30 pt \le 2 s_{2, d} \ell^{2d -1}  \max_{\bolda \in \partial_-^{{\blam}_2}\pa{{\blam}_1} }  \abs{  \pa{H_{{\blam}_1} - z}^{-1}  \pa{\boldu, \bolda}} \abs{\pa{H_{{\blam}_2} - z}^{-1} \pa{\boldb_1, \boldv}    }  \notag
\end{align}
for some $\boldb_1 \in  \partial_+^{{\blam}_2}\pa{{\blam}_1}$.

\begin{definition}
Let ${\blam_S} ={\blam}^{(2)}_{\paS; \bL} (\x)$   with $\ell = \min \pa{L_1, L_2}$, and $E\in \R$.
\begin{enumerate}
\item Let $s > 0$. Then  $\boldlambda_S$ is  $(E,s)$-suitably resonant if and only if
\beq
\dist \Bl( \sigma \bigl( H_{\boldlambda_S} \bigr) , E \Br) < \ell^{-s}.
\eeq
  Otherwise, $\boldlambda_S$ is  $(E,s)$-suitably nonresonant.

 \item Let $\beta \in (0, \, 1)$.  Then $\boldlambda_S$ is  $(E,\beta)$-resonant if and only if
 \beq
 \dist \Bl( \sigma \bigl( H_{\boldlambda_S} \bigr) , E \Br) < \tfrac{1}{2}  e^{-\ell^{\beta}} .
 \eeq
  Otherwise, $\boldlambda_S$ is   $(E,\beta)$-nonresonant.
\end{enumerate}

\end{definition}

Let $\boldlambda ={\blam}_{\paS; L}^{\pat} (\x)$  be a symmetrized two-particle box. We set
\begin{align} \label{suitc0}
\Xi_{L, \ell}(\x) = \set{ \boldx + \pa{\tfrac \ell 3 +1}  \Z^{2d}}\cap \set{\y \in \R^{2d};\,\, \norm{\y-\x}_\infty \le \tfrac L 2 - \ell    }.
\end{align}
 The symmetrized $\ell$-suitable partial cover of ${\blam}$ (or the $\ell$-suitable partial cover of ${\blam}$ for short) is the collection of symmetrized two-particle boxes
\beq \label{suitc}
\cC_{L, \ell}(\x) = \set{ {\blam}^{(2)}_{\pas; \ell} (\boldy); \;  \y \in \Xi_{L, \ell}(\x)}.
\eeq
Note that  we have
\begin{gather}\label{suitc1}
{\blam}^{(2)}_{\pas; \ell} (\boldy) \subset {\blam} \qtx{and}  \partial_{-}{\blam} \, \cap \, {\blam}^{(2)}_{\pas; \ell} (\boldy)   = \emptyset  \qtx{for all} \y \in \Xi_{L, \ell}(\x),\\ \# \cC_{L, \ell}(\x) < \pa{ 2 \pa{ 3 \tfrac{L}{\ell}  + 1} }^{2d}. \notag
 \end{gather}
Moreover, for every $\boldu \in {\blam}^{(2)}_{\pas; L - 2\ell}(\x)$ there exists $\boldy \in \Xi_{L, \ell}(\x)$ such that $\boldu \in {\blam}^{(2)}_{\pas; \frac{\ell}{3}}(\y)$.

\begin{remark}
Elements in $\Xi_{L, \ell}(\x)$ will be referred to as cells. Two cells, $\bolda, \boldb \in \Xi_{L, \ell}(\x)$, are  neighbors if and only if $\dist_\infty \pa{\bolda, \boldb} = \tfrac{\ell}{3} + 1$. Moreover, by construction, if $\boldy \in \Xi_{L, \ell}(\x)$ is sufficiently far away from the boundary of ${\blam}_{\paS; L}^{\pat} (\x)$, then for every $\boldu \in \partial_{+} \pa{ {\blam}^{(2)}_{\pas, \ell}(\boldy)}$, we have $\boldu \in {\blam}^{(2)}_{\pas, \ell}(\bolda)$ for some $\bolda \in \Xi_{L, \ell}(\x)$ that is a neighbor to $\boldy$.
\end{remark}

 \begin{definition} \label{PIFIbox}
Let ${\blam} \subsetneq \Z^{2d}$. ${\blam}$
is said to be  non-interactive   if and only if for every $\y = \pa{y_1, y_2} \in {\blam}   $, we have
\beq
\norm{y_1 - y_2} > r_0
\eeq
Otherwise, it is said to be    interactive .
\end{definition}

 \begin{proposition} \label{PIcriteria}
 Let ${\blam}_1 =  \Lambda_{L}(x_1) \times \Lambda_{\ell}(x_2 )$ be a two-particle rectangle, ${\blam}   = \blam_1 \; \cup \; \pi \pa{\blam_1}$ be a symmetrized two-particle rectangle   with $L \ge \ell$. If   $\norm{x_1 - x_2} > L + r_0$, then we have the following:
 \begin{enumerate}
 \item $ {\blam} $ is non-interactive,
 \item $\dist \pa{\Lambda_{L}(x_1) , \Lambda_{\ell}(x_2)} > r_0$,
 \item $H_{{\blam}_1 }  = H_{\Lambda_{L} (x_1)} \otimes I + I \otimes H_{\Lambda_{\ell} (x_2)} $,
 \item $H_{\pi \pa{{\blam}_1}}  = H_{\Lambda_{\ell} (x_2)} \otimes I + I \otimes H_{\Lambda_{L} (x_1)} $,
 \item $\sigma \pa{ H_{{\blam}_1} } = \sigma \pa{ H_{ \pi \pa{{\blam}_1}} } = \sigma \pa{ H_{\Lambda_{L} (x_1)} } +  \, \sigma \pa{ H_{\Lambda_{\ell} (x_2)} }$,

 \item $\dist \pa{ {\blam}_1, \pi \pa{{\blam}_1}   } > 1$

  \item $H_{{\blam}} =  H_{{\blam}_1}   \oplus H_{\pi \pa{{\blam}_1}}  $,

 \item $\sigma \pa{H_{ \blam}} = \sigma \pa{ H_{\Lambda_{L} (x_1)} } + \,\sigma \pa{ H_{\Lambda_{\ell} (x_2)} }$,

 \item for every $E \in \sigma \pa{H_{{\blam}}   }$,
\beq \label{PIequality}
\pa{  H_{{\blam} } - E }^{-1} = \pa{ H_{{\blam}_1}  - E }^{-1} \oplus \pa{ H_{ \pi\pa{{\blam}_1}}  - E }^{-1}.
\eeq

\end{enumerate}
 \end{proposition}

 \begin{proof}
It's clear that $(i) \Rightarrow (ii)$, $(vi) \Rightarrow (vii)$, and $(vii) \Rightarrow (viii) \text{ and } (ix)$. By \cite[Lemma 2.7]{KlN1}, $(ii) \Rightarrow (iii), (iv), \text{ and }  (v)$. $(ii) \Rightarrow (vi)$ since $r_0 \ge 1$.
\end{proof}

\begin{remark}
Note that a two-particle box ${\blam}^{(2)}_{L}(\x)$ is  non-interactive if and only if  ${\blam}^{(2)}_{\pas; L}(\x)$ is non-interactive. Hence, if  ${\blam}^{(2)}_{L}(\x)$ is  non-interactive, then ${\blam}^{(2)}_{\pas; L}(\x)$ satisfies (ii)-(ix).
\end{remark}

We recall   \cite[Lemma 3.9]{KlN2}, which is an important ingredient for the two-particle bootstrap MSA (see also \cite[Lemma 2.8]{KlN1}).

\begin{lemma}\label{PIsuit}
  Let ${\blam}^{(2)}_{L}(\boldx) $ be a two-particle box and ${\blam}^{(2)}_{\paS; L}(\x)  $ be the corresponding symmetrized two-particle box. Let  $E \leq E^{(2)}$ and $E^{(1)}>0$ such that $E^{(1)}>E^{(2)}$. If $\norm{x_1 - x_2} > L + r_0 $ and  $L$ is sufficiently large,

\begin{enumerate}

\item   Given $\theta> 2d+2 $, suppose  ${\La_{L}(x_1)}$ is $(\theta, \,E - \mu  )$-suitable  for every $\mu \in \sigma \pa{H_{\La_{L}(x_2)}}\cap (-\infty,E^{(1)}]   $ and  ${\La_{L}(x_2)}$ is $(\theta, \,E - \lambda  )$-suitable  for every $\lambda \in \sigma \pa{H_{\La_{L}(x_1)}}\cap (-\infty,E^{(1)}]  $.  Then ${\blam}^{(2)}_{\paS; L}(\x) $ is $\left(\tfrac{\theta}{2}, \, E \right)$-suitable.

\item   Given $0<m< \log\pa{\frac{E^{(1)}-E^{(2)}}{4d}+1}$, suppose  ${\La_{L}(x_1)}$ is $(m,\, E - \mu)$-regular  for every $\mu \in \sigma \pa{H_{\La_{L}(x_2)}}\cap (-\infty,E^{(1)}]   $ and  ${\La_{L}(x_2)}$ is $(m,\, E - \lambda)$-regular  for every $\lambda \in \sigma \pa{H_{\La_{L}(x_1)}} \cap (-\infty,E^{(1)}]  $. Then ${\blam}^{(2)}_{\paS; L}(\x) $ is $\pa{ m-\tfrac{6(d+1) \log 2L}{L}, \, E }$-regular.

\item   Given $0 < \zeta^{\prime} < \zeta < 1$, suppose  ${\La_{L}(x_1)}$ is $(\zeta,\, E - \mu )$-SES  for every $\mu \in \sigma \pa{H_{\La_{L}(x_2)}} \cap (-\infty,E^{(1)}]   $ and  ${\La_{L}(x_2)}$ is $(\zeta,\, E - \lambda)$-SES  for every $\lambda \in \sigma \pa{H_{\La_{L}(x_1)}}\cap (-\infty,E^{(1)}]   $.  Then ${\blam}^{(2)}_{\paS; L}(\x) $ is  $\left(\zeta^{\prime}, \, E \right)$-SES.
\end{enumerate}
\end{lemma}

\begin{proof} We prove (ii), the proofs of parts (i) and (iii) being similar. We begin by  showing that ${\blam}^{(2)}_{L}(\x)$ is $\left(m-\tfrac{6(d+1) \log 2L}{L}, \, E \right)$-regular.

Given $\boldu \in {\blam}^{(2)}_{L/3}(\x)$ and $\y\in \partial_{-} {\blam}^{(2)}_{L}(\x) $, $\norm{\boldu-\y}> \frac{L}{6}$. Then either $\norm{u_1-y_1}> \frac{L}{6}$ or  $\norm{u_2-y_2}> \frac{L}{6}$. Without loss of generality we assume the latter.
We write $\sigma_1 :=\sigma \pa{H_{\La_{L}(x_1)}}\cap (-\infty,E^{(1)}]$ and $\sigma_1^c:=\sigma \pa{H_{\La_{L}(x_1)}}\cap (E^{(1)}, +\infty)$. We use the hypothesis, the Combes-Thomas estimate from Theorem \ref{thmCT} (it is enough to take $\eps=1/2$ there) and the fact that $\dist\pa{E-\mu,\sigma \pa{H_{\La_{L}(x_1)}}}> E^{(1)}-E^{(2)}$ for $\mu\in\sigma_1^c$ and $E\leq E^{(2)}$, to show that
\begin{align}
  \abs{ G_{{\blam}} (E;\,\boldu,\,\y) } & \leq \sum_{ \mu \in \sigma_1}
\abs{ G_{ \Lambda_{L}( x_2)} (E - \mu;\,u_2,\,y_2) }	\notag	 \\
& \quad \quad + \sum_{ \mu \in \sigma_1^c}
 \abs{ G_{ \Lambda_{L}( x_2)} (E - \mu;\,u_2,\,y_2) }	\notag	 \\
&   \le L^d e^{-m\norm{u_2-y_2}} + \frac{2L^d}{E^{(1)}-E^{(2)}} e^{-\log\pa{\frac{E^{(1)}-E^{(2)}}{4d}+1} \norm{u_2-y_2}}\notag\\
&  \le \; L^{d+1}\pa{e^{-m\norm{u_2-y_2}}+ e^{-\log\pa{\frac{E^{(1)}-E^{(2)}}{4d}+1} \norm{u_2-y_2}}} .
\end{align}
for $L$ large enough. Taking $ m< \log\pa{\frac{E^{(1)}-E^{(2)}}{4d}+1}$ yields that, for $L$ large enough, depending on $E^{(1)}-E^{(2)}$ and $d$, the box ${ {\blam}^{(2)}_{L}(\boldx) }$ is $\left(m-\tfrac{6(d+1) \log 2L}{L}, \, E \right)$-regular. Similarly, one can show the same holds for ${\blam}^{(2)}_{L}(\pi(\x))$. By Equation \eqref{PIequality}, we conclude that ${\blam}^{(2)}_{\paS; L}(\x) $ is $\pa{ m-\tfrac{6(d+1) \log 2L}{L}, \, E }$-regular.
\end{proof}

\begin{proposition} \label{parsep}
Given a pair of symmetrized two-particle boxes, $ {\blam}^{(2)}_{\paS; L} (\x)$ and $ {\blam}^{(2)}_{\paS; L} (\y)$   such that  $ d_{S} \pa{\x, \y } > L$, then the pair is partially separated.
\end{proposition}

\begin{proof}
Since $\norm{\x - \y} > L$ and $\norm{\pi (\x) - \y} > L$, we have
\begin{align}\label{max}
& \max \set{ \norm{x_1 - y_1}, \norm{x_2 - y_2}  } > L , \qtx{and} \\
&\max \set{ \norm{x_1 - y_2}, \norm{x_2 - y_1}  } >  L . \notag
\end{align}
There are four cases to consider here. If we take $\norm{x_1 - y_1} > L$ and $\norm{x_2 - y_1} >  L$, then we have that the pair is partially separated since
\beq
\Pi_1 {\blam}^{(2)}_{ L} (\y) \bigcap   {\blam}^{(2)}_{\paS; L} (\x) = \emptyset.
\eeq
The other three cases are handled in the same way.
\end{proof}

\begin{definition} \label{Ldistant}
Given a pair of symmetrized two-particle boxes, $ {\blam}^{(2)}_{\paS; L} (\x)$ and $ {\blam}^{(2)}_{\paS; L} (\y)$. We say that the pair is $L$-distant if and only if
\beq
\dist_{S} \pa{\x, \y } > 8L.
\eeq
\end{definition}

It follows from Proposition \ref{parsep} that $L$-distant automatically implies partially separated.

\begin{proposition}\label{fullysep}
Let ${\blam}^{(2)}_{\paS; L} (\x) $ and  $ {\blam}^{(2)}_{\paS; L} (\y) $ be a pair of symmetrized two-particle boxes.
 If the pair is interactive and $L$-distant, then ${\blam}^{(2)}_{\paS; L} (\x) $ and ${\blam}^{(2)}_{\paS; L} (\y)  $ are fully separated, provided $L$ is sufficiently large.

\end{proposition}

\begin{proof}
Since ${\blam}^{(2)}_{\paS; L} (\x) $ and ${\blam}^{(2)}_{\paS; L} (\y)  $ are interactive, thus there exists $\bolda \in {\blam}^{(2)}_{\paS; L} (\x)$ and $\boldb \in {\blam}^{(2)}_{\paS; L} (\y)$ such that
\beq
\norm{a_1 - a_2 } \le r_0, \qtx{and} \norm{b_1 - b_2} \le r_0.
\eeq
Thus, we have
\beq\label{inter}
\dist \pa{\La_{L}(x_1), \La_{L}(x_2)} \le r_0 \qtx{and} \dist \pa{\La_{L}(y_1), \La_{L}(y_2)} \le r_0.
\eeq
Since $d_{S} \pa{\x, \y } > 8L  $, we can proceed as in \eqref{max} and see that, in particular
\beq\notag  \max \set{ \norm{x_1 - y_1}, \norm{x_2 - y_2}  } > 8L.\eeq

Thus, $ \norm{x_1 - y_1} > 8L$ or $ \norm{x_2 - y_2} > 8L$. Without loss of generality, let us  assume $ \norm{x_1 - y_1} > 8L$, then
\beq
\La_{4L}(x_1) \bigcap \La_{4L}(y_1) = \emptyset.
\eeq
Moreover, \eqref{inter}
implies
\begin{align}
& \La_{L}(x_1) \bigcup \La_{L}(x_2) \subseteq \La_{3L + 2r_0}(x_1) \subset \La_{4L}(x_1)\; \text{ and } \notag\\
& \hskip 30 pt \La_{L}(y_1) \bigcup \La_{L}(y_2) \subseteq \La_{3L + 2r_0}(y_1) \subset \La_{4L}(y_1),
\end{align}
provided $L$ is sufficiently large. Therefore,
\beq
 \Pi {\blam}^{(2)}_{\paS; L} (\x)  \bigcap  \Pi {\blam}^{(2)}_{\paS; L} (\y) = \emptyset.
 \eeq
\end{proof}

\subsection{The Bootstrap Multiscale Analysis}\label{S:MSA33}

The proof of  Theorem~\ref{MSAthm} is similar to the proof of \cite[Theorem 1.6]{KlN2}.
It uses the corresponding result for one-particle boxes from \cite{GK1} and the definitions and results for symmetrized boxes introduced in the previous sections. We will state the steps needed for the proof and refer to \cite{KlN1} and \cite{KlN2} for details.

\subsubsection{The initial step for the MSA at the bottom of the spectrum}
That the hypotheses of Theorem \ref{MSAthm} are satisfied at the bottom of the spectrum can be proven in the same way as  \cite[Theorem 4.1]{KlN2} using Theorem~\ref{thmCT}.

\subsubsection{Consequences of the one-particle case}

For the one-particle model, note that symmetrized boxes are usual boxes in the distance $\dist_\infty$. In this case we know from \cite[Theorem 3.4]{GK1} that there exists $E^{(1)}>0$ such that, given $E^{(2)}<E^{(1)}$, for every  $\tau \in (0, \, 1)$ there is a length scale $L_{\tau} $,  $\delta_{\tau}>0$, and $ 0< m_\tau^{*} < \log\pa{\frac{E^{(1)}-E^{(2)}}{4d}+1}$, such that the following hold for all $E\leq E^{(1)}$  :
\begin{enumerate}[i)]
\item   For all $L \geq L_{\tau}$ and $a \in  \R^{d}$ we have
 \begin{align} \label{eq001}
\P \Bigl\{   \Lambda_{L}(a) \,\, \text{is  $ \left (m^{*}_\tau, \,E \right )$-nonregular} \Bigr\}    \leq e^{-L^{\tau}}.
\end{align}

\item Let $I(E)=[E-\delta_{\tau}, E+\delta_{\tau}] \subset (-\infty,E^{(1)}]$.  For all $L \geq L_{\tau}$  and all pairs  of
 disjoint one-particle boxes $\Lambda_{L}(a)$ and $\Lambda_{L}(b)$ we have
\begin{align} \notag
&\P \Bigl\{ \exists \, E^\pr \in I(E) \; \text{so both} \;  \Lambda_{L}(a) \text{ and}\: \Lambda_{L}(b) \; \text{are}\; \left ({m_\tau^*}, \,E^\pr \right )\text{-nonregular} \Bigr\} \\
& \hspace{50pt} \leq e^{-L^{\tau}} .  \label{eq002}
\end{align}
\end{enumerate}

\begin{remark}\label{lktocont}The result of \cite{GK1} holds for a sequence of scales.  It holds for all sufficiently large scales using \cite[ Lemma 3.16]{GKber} as in \cite[Theorem 3.21]{KlN1}.
\end{remark}

 This result yields probability estimates for non-interactive symmetrized two-particle boxes.

\begin{lemma} \label{PINSl}
Let ${\blam}^{(2)}_{\ell}(\x) = \Lambda_{\ell}({x}_{1}) \times  \Lambda_{\ell}({x}_{2} )$ be a non-interactive two-particle box and $\tau \in (0,1)$. Then for $\ell$ large enough, depending on $E^{(1)},E^{(2)},d,\tau$, and for all $E\leq E^{(2)}$  we have
\begin{align}
&\P  \Bigl\{ {\blam}^{(2)}_{\paS; \ell}(\x) \;\text{is} \; \left( m^*_{\tau}(\ell),E \right)\text{-nonregular} \Bigr\} \leq \ell^{2d} e^{-\ell^{\tau}},\sqtx{with}m^*_{\tau}(\ell)=m^*_{\tau}- \tfrac{6(d+1)\log 2\ell}{\ell}, \notag\\
&\P  \Bigl\{ {\blam}^{(2)}_{\paS; \ell}(\x) \,\, \text{is} \,\, \left( \theta,\,E \right)\text{-nonsuitable} \Bigr\} \leq \ell^{2d} e^{-\ell^{\tau}} \text{for} \; \; \theta < \tfrac \ell {\log \ell} \tfrac{ m^*_{\tau}(\ell)}{2},\\ \notag
&\P  \Bigl\{ {\blam}^{(2)}_{\pas; \ell}(\x) \,\, \text{is} \,\, \left(\tau,\,E \right)\text{-nonSES} \Bigr\} \leq \ell^{2d} e^{-\ell^{\tau}}.
\end{align}
\end{lemma}

\begin{proof} Since both ${\blam}^{(2)}_{\ell}(\x)$ and ${\blam}^{(2)}_{\ell}(\pi(\x))$ are non-interactive, the proof is a direct consequence of \cite[Lemma 5.1]{KlN2} adapted to the discrete setting and Lemma \ref{PIsuit}.
\end{proof}

In what follows, we fix $\zeta, \tau, \beta, \zeta_0, \zeta_1, \zeta_2, \gamma$ such that
\begin{gather} \label{constant}
0 < \zeta < \tau < 1, \qtx{} \gamma > 1, \\
\zeta < \zeta_2 < \gamma \zeta_2 < \zeta_1 < \gamma \zeta_1 < \beta <\zeta_0 < \tau \qtx{with} \zeta \gamma^2 < \zeta_2. \notag
\end{gather}
To aliviate the notation, in what follows we will omit the upperscript $(2)$ from the notation of two-particle boxes and
write simply $\blam_L(\x),\blam_{S;L}(\x)$.

\subsubsection{The first Multiscale Analysis for symmetrized two-particle boxes}

\begin{proposition}\label{1stMSA} Let $E\leq E^{(2)}$,  $\theta>16d$, $0<p< \theta-8d +2$,
$Y \geq 66$, and $0<  p_0 < ( 6Y + 2  ) ^{-4d }$.   If for some sufficiently large $L_0 \in  \N$  we have
\beq
\sup_{\boldx \in \R^{2d}} \P \Bigl\{ \mathbf{\Lambda}_{\paS; L_0}(\boldx) \, \,  \text{is}\,\,(\theta, \,E)  \text{-nonsuitable} \Bigr\} \leq p_0,
\eeq
then, setting $L_{k+1} = Y\,L_k,$ for $k = 0, 1, 2, ...,$ there exists $K_0 \in \N$ such that for every $k \geq K_0$ we have
\beq
\sup_{\boldx \in \R^{2d}} \P \Bigl\{ \mathbf{\Lambda}_{\paS;L_k} (\boldx)    \,\,\text{is}\,\,(\theta, \,E)  \text{-nonsuitable} \Bigr\} \leq L_{k}^{-p}.
\eeq
\end{proposition}

The proof relies on the following lemma, in the same way \cite[Proposition 3.2]{KlN1} relies on \cite[Lemma 3.3]{KlN1}.

\begin{lemma} \label{deterministicMSA1}
Let $E\in\R$, $s>0$,   $\theta >  4d - 2 + s$,  $J \in \N$, $Y\geq 10+ 56J$, and $L  = Y \ell$, and consider   a symmetrized two-particle box ${\blam} := {\blam}_{{S}; L}(\x)$ with the usual $\ell$-suitable partial cover. Suppose
\begin{enumerate}[i)]
\item  ${\blam}$ is $(E,s)$-suitably nonresonant.
\item There exist at most $J$ pairwise $\ell$-distant symmetrized two-particle boxes in the suitable partial cover $\cC_{L, \ell}(\x)$ that are $(\theta, E)$-nonsuitable.
\item Every symmetrized two-particle box with center belonging to  $\Xi_{L, \ell}(\x)$ of side length  $j\pa{8\ell+1}$ with $j \in \set{1, \dots, J}$ is $(E,s)$-suitably nonresonant.
\end{enumerate}
Then ${\blam}$ is $(\theta, E)$-suitable for $L$ sufficiently large.
\end{lemma}

\begin{proof}
Without loss of generality, let us assume $  {\blam}_{{S}; \ell}(\bolda_1), \dots, {\blam}_{{S}; \ell}(\bolda_J)$ are the  $J$ pairwise $\ell$-distant symmetrized two-particle boxes in ${\blam}$ that are $(\theta, E)$-nonsuitable. This implies that if ${\blam}_{{S}; \ell}(\boldb)$ is   $\ell$-distant from ${\blam}_{{S}; \ell}(\bolda_j)$ for every $j \in \set{ 1, \dots, J}$, and ${\blam}_{{S}; \ell}(\boldb) \in \cC_{L, \ell}(\x)$, then ${\blam}_{{S}; \ell}(\boldb)$ must be $(\theta, E)$-suitable.

Let us denote $\cT = \set{ \bolda_1, \dots,  \bolda_J}$, and
\beq
{\blam}_0 = \bigcup_{i = 1, \dots, J} {\blam}_{{S}; 8\ell}(\bolda_i).
\eeq
We can separate the set ${\blam}_0$ into clusters $P_1, \dots, P_r$, which will be referred to as bad clusters, so that
\begin{enumerate}[i)]
\item for $i =1, \dots, r$, $P_i \subset {\blam}$, and each $P_i$ is a symmetrized two-particle box of length $  t_i(8\ell + 1) \le 9 t_i \ell $ provided  $t_i$ is the maximum number of  elements in $\cT$ that belong to $P_i$,
\item $\dist \pa{P_i, P_j} > 1$ for $i \neq j$,
\item $\bigcup_{i = 1, \dots, J} {\blam}_{{S}; \ell}(\bolda_i) \subset \bigcup_{i = 1, \dots, r} P_r$,
\item $\sum_{i = 1, \dots, r} 9 t_i \ell \le 9J\ell$,
\item if $\boldb \notin P_i$ for every $i = 1, \dots, r$, and $\dist_{S}(\x, \boldb) \le \tfrac{L}{2} - \ell$, then there exists $\boldu \in \Xi_{L, \ell}(\x)$ such that $\boldb \in {\blam}_{\pas; \ell}(\boldu)$ and ${\blam}_{\pas; \ell}(\boldu)$ is $\PET$-suitable.
   \end{enumerate}

Let us now fix $\boldy \in \partial_{-} {\blam}$. Given  $\boldu \in \blam_{S,\frac{L}{3}} (\x) $ with $\boldu \in {\blam}_{\pas; \alphal}(\bolda)$ for some $\bolda \in \Xi_{L, \ell}(\x)$, (note that $\boldy \notin {\blam}_{\paS; \ell}(\bolda)$ by construction), we have the following cases to consider:
\begin{enumerate}[a)]
\item if $\bolda \notin P_i$ for every $i \in \set{1, \dots, r}$, then
  ${\blam}_{\paS; \ell}(\bolda)$ is $\PET$-suitable. Setting ${\blam}_1 = {\blam}_{\paS; \ell}(\bolda)$ and using the geometric resolvent identity \eqref{resolventine}, we obtain
\begin{align} \label{goodjump}
& \abs{  \pa{H_{{\blam}} - E }^{-1} (\boldu, \boldy) } \notag \\
& \hskip 20pt \le 2 s_{2, d} \ell^{2d -1}  \max_{\boldb \in \partial_-^{{\blam}} \pa{{\blam}_{1} }}  \abs{  \pa{H_{{\blam}_{1} } - E}^{-1}  \pa{\boldu, \boldb} }  \abs{\pa{H_{{\blam}} - E}^{-1} \pa{\boldb_1, \boldy}    }   \notag \\
& \hskip 20pt \le 2 s_{2, d} \ell^{2d -1} \ell^{-\theta} \abs{\pa{H_{{\blam}} - E}^{-1} \pa{\boldb_1, \boldy}    }
\end{align}
with $\boldb_1 \in  \partial_+^{{\blam}}\pa{{\blam}_1}$. Then by construction, $\boldb_1 \in {\blam}_{\pas; \ell}(\boldv)$, where $\boldv \in \Xi_{L , \ell}(\x)$ is a neighboring cell of $\bolda$.

\item if $\bolda \in P_i$ for some $i \in \set{1, \dots, r}$, and $\boldy \notin P_i$, then, using the fact that $P_i$ is   $(E,s)$-suitably nonresonant, we have
\begin{align}
& \abs{  \pa{H_{{\blam}} - E }^{-1} (\boldu, \boldy) } \notag \\
& \hskip 20pt \le 2 s_{2, d} \ell^{2d -1}  \max_{\boldv \in \partial_-^{{\blam}} \pa{ P_i }}  \abs{  \pa{H_{P_i}  - E}^{-1}  \pa{\boldu, \boldv} }  \abs{\pa{H_{{\blam}} - E}^{-1} \pa{\boldv_1, \boldy}    }   \notag \\
& \hskip 20pt \le 2 s_{2, d} \ell^{2d -1} (Y \, \ell)^{s} \abs{\pa{H_{{\blam}} - E}^{-1} \pa{\boldv_1, \boldy}    }
\end{align}
with $\boldv_1 \in  \partial_+^{{\blam}}\pa{P_i}$. If there exists $\boldb \in \Xi_{L, \ell}(\x)$ such that $\boldv_1 \in {\blam}_{\pas; \alphal}(\boldb)$, $\y \notin {\blam}_{\pas; \ell}(\boldb)$, and ${\blam}_{\pas; \ell}(\boldb)$ is $\PET$-suitable, then
we can repeat Equation \eqref{goodjump} with $\boldv_1$ replacing $\boldu$. Then
\begin{align} \label{bad jump}
& \abs{  \pa{H_{{\blam}} - E }^{-1} (\boldu, \boldy) } \notag \\
& \hskip 10pt \le 2 s_{2, d} \ell^{2d -1} (Y \, \ell)^{s} \abs{\pa{H_{{\blam}} - E}^{-1} \pa{\boldv_1, \boldy}    } \notag\\
&\hskip 20pt \le 2 s_{2, d} \ell^{2d -1} (Y \, \ell)^{s} 2 s_{2, d} \ell^{2d -1} \ell^{-\theta} \abs{\pa{H_{{\blam}} - E}^{-1} \pa{\boldv_2, \boldy}    } \notag \\
&\hskip 30pt \le \abs{\pa{H_{{\blam}} - E}^{-1} \pa{\boldv_2, \boldy}    },
\end{align}
where $\boldv_2 \in  \partial_+^{{\blam}}\pa{{\blam}_{\pas; \ell}(\boldb)}$, provided $2 s_{2, d} \ell^{2d -1} (Y \, \ell)^{s} 2 s_{2, d} \ell^{2d -1} \ell^{-\theta} \le 1$. Since by hypothesis, $\theta > 4d +s -2$, this can be achieved if  $\ell$ is sufficiently large  depending on $s_{2,d},Y$ and $s$. Moreover, $\boldv_2$ belongs to a neighboring cell of $\boldb$. If such $\boldb$ does not exist, then we can always conclude
\beq
\abs{  \pa{H_{{\blam}} - E }^{-1} (\boldu, \boldy) } \le L^{s}.
\eeq

\end{enumerate}
We  will estimate $\abs{  \pa{H_{{\blam}} - E }^{-1} (\boldu, \boldy) }$ with $\boldu \in {\blam}_{\pas; \alphaL} (\x)$. We first note that we can start with $\boldu \in {\blam}_{\paS; \alphaL}(\x)$ and apply either procedure (a) or (b) repeatedly.  Since $\pa{ \tfrac{L}{2} - \ell - \tfrac{L}{6} } \pa{ \tfrac{\ell}{3} +1}^{-1} = \pa{ \tfrac{L}{3} - \ell   } \pa{ \tfrac{\ell}{3} +1}^{-1} \ge \tfrac{L}{\ell +3} - 3$, for our purpose, the shortest path starting from $\boldu$ to the boundary of ${\blam}$ has at least  $\tfrac{L}{\ell +3} - 3$ cells.
With $L = Y \ell$, thus we get
\begin{align} \label{hopping1}
\tfrac{L}{\ell +3} - 3 =Y \pa{ \tfrac{\ell}{\ell +3} } - 3 > \tfrac{Y}{2} - 3.
\end{align}
Then
\beq
\abs{  \pa{H_{{\blam}} - E }^{-1} (\boldu, \boldy) } \le \pa{ 2 s_{2, d} \ell^{2d -1} \ell^{-\theta}   }  ^{N(Y)}  L^{s},
\eeq
where $N(Y)$ is the number of cells for which we are able to perform (a) without using the result for the control of a bad region. Having to account for the cells  where we do not get anything due to the bad regions, which is, in the worst case, $9J\ell \pa{ \tfrac{\ell}{3} +1}^{-1} + J \le 9J\ell \pa{ \tfrac{\ell}{3}}^{-1} + J = 28J$,  equation \eqref{hopping1} gives us
\beq
N(Y) \ge \tfrac{Y}{2} - 3 - 28J.
\eeq
Our goal is to have $\abs{  \pa{H_{{\blam}} - E }^{-1} (\boldx, \boldy) } \le L^{-\theta}$, so we would like
 \beq
  \pa{2 s_{2, d} \ell^{2d -1} \ell^{-\theta} } ^{ \tfrac{Y}{2} - 3 - 28J}  \pa{Y \ell}^{s} \le (Y \ell)^{-\theta}.
  \eeq
   This can be achieved  for $\ell$ sufficiently large if
 \beq \label{useful1}
 \pa{2d - 1 - \theta}  \pa{\tfrac{Y}{2} - 3 - 28J} + s + \theta < 0,
 \eeq
provided $\ell$ is sufficiently large  depending on $s_{2,d},Y$ and $s$.   Since $\theta >  4d - 2 + s$, it follows that  \eqref{useful1} is true if we have $\tfrac{Y}{2} - 3 - 28J \ge 2$, that is, $Y\geq 10+ 56J$.
\end{proof}

\begin{proof}[Proof of Proposition \ref{1stMSA}]  Using   $0<p< \theta-8d +2$,   we fix $s>0$ such that $4d +p <s< \theta -4d +2$.

Given a scale $L$, we set
\beq
p_L= \sup_{\boldx \in \R^{2d}} \P \Bigl\{ \mathbf{\Lambda}_{\paS; L}(\boldx) \, \,  \text{is}\,\,(\theta, \,E)  \text{-nonsuitable} \Bigr\} .
\eeq
Let ${\blam} = {\blam}_{\paS; L} (\boldx)$ be a two-particle symmetrized box with an $\ell$-suitable cover $\cC_{L, \ell}(\x)$, where $L=Y\ell$. We begin by
defining several events:   \\
$\cE=  \Bl\{{\blam} \; \text{ is } (\theta, E)\text{-nonsuitable}\Br\}$;\\
$\cA$ is the event that there exists a non-interactive box, ${\blam}_{\paS; \ell} (\boldv) \in \cC_{L, \ell}(\x)$, that  is  $(\theta, E)$-nonsuitable, \\
$\cW_J$ is the event that ${\blam}$ is $(E,s)$-suitably nonresonant and ${\blam}_{\pas, j(8\ell+1)}(\bolda)$  is $(E,s)$-suitably nonresonant for every $\bolda \in \Xi_{L, \ell}(\x)$ and every $j \in \set{1, \dots, J}$, and\\
$\cF_J$ is the event that there are at most $J$ pairwise $\ell$-distant symmetrized two-particle boxes in $\cC_{L, \ell}(\x)$ that are $(\theta, E)$-nonsuitable.

 We take $Y\geq 10 + 56J$, with $J\in\N$ to be determined later. Then, by Proposition \ref{deterministicMSA1} we have
\beq
\P\set{\cE} \le \P\set{\cW_J^c}  + \P\set{\cF_{J}^c} \le  \P\set{\cW_J^c}  + \P\set{\cF_{J}^c\cap \cA^c}+\P\set{\cA}.
\eeq

Note that
\beq
\cW_J^c \subset \bigcup_{\substack{ \boldlambda_{\paS; t} ( \y ) = {\blam} \text{ or } \boldy \in \Xi_{L, \ell}(\x)      \\ t \in \set{ j(8\ell+1) \; | \; j = 1, \dots, J   }    }   }
\Bl\{ \dist \left ( \sigma \pa{ H_{\boldlambda_{\paS; t} ( \y )} }; E \right) \leq t^{-s} \Br\}.
\eeq
With our choice of $s$, Theorem \ref{t:wesharp} implies
\begin{align*}
&\P\set{\cW_J^c}  \leq  \pa{ J (2 \pa{3 \tfrac{L}{\ell} + 1}  )^{2d}   }     \pa{8 \norm{\rho}_{\infty}  \ell^{-s}  L^{2d}}   \\
&\hskip 30pt = J 2^{2d} (3Y + 1)^{2d} \norm{\rho}_{\infty} Y^{4d} \ell^{4d - s }      \leq \tfrac{1}{4} (Y\ell)^{-p} = \tfrac{1}{4} L^{-p},
\end{align*}
provided $4d + p < s$ and $\ell$ is sufficiently large. On the other hand,
 Lemma \ref{PINSl} yields, for $\ell$ is sufficiently large depending on $Y,d,p,\tau$, {$ E^{(1)},E^{(2)}$},
\[\P\set{\cA} \leq (Y\ell)^{2d}\,\ell^{2d}e^{-\ell^{\tau}} \leq \tfrac{1}{4} L^{-p}.\]

Hence
 \beq \label{part1keyeq}
 \P \Bl\{{\blam}_{\paS; L}(\x) \; \text{ is } (\theta, E)\text{-nonsuitable}\Br\} \leq \tfrac{1}{2}L^{-p} + \P \set{\cF_{J}^c\cap \cA^c}.
 \eeq
To bound $\P \set{\cF_{J}^c\cap \cA^c}$, we note that $\bom \in \cF_{J}^c\cap \cA^c$ implies there are $J+1$ interactive pairwise $\ell$-distant boxes in $\cC_{L, \ell}(\x)$ that are $(\theta, E)$-nonsuitable.
Hence, using  Proposition \ref{fullysep},  and recalling \eqref{suitc}-\eqref{suitc1}, we get

\begin{align}
& \P \set{\cF_{J}^c\cap \cA^c} \leq   \tfrac 1 2   \pa{ p_{\ell} } ^{J+1} ( 6Y + 2  ) ^{2d (J+1)} =  \tfrac 1 2  \ell^{-p(J+1)} ( 6Y + 2  ) ^{2d (J+1)}.
\end{align}
Our goal is to have $ \P \set{\cF_{J}^c\cap \cA^c} \le \frac{1}{2}L^{-p}$, which means we just need to require
\beq \label{Jnecessary1}
 \tfrac 1 2  \pa{\ell^{-p} ( 6Y + 2  ) ^{2d } }^{J+1}  \le \tfrac{1}{2} (Y \ell )^{-p}.
\eeq
 This can be achieved if $ J \ge 1 $, provided $\ell$ is sufficiently large.

Next, our goal is to show that for sufficiently large $L_0$, letting $L_k = Y L_{k-1}$, then there must exists a $K_0$ such that
\beq
p_{K_0} \leq L_{K_0}^{-p}.
\eeq

We will proceed by seeing what happens when $K_0 \neq 0, 1, 2, ...$.  From equation \eqref{part1keyeq}, taking $\ell$ to be large enough, then
\begin{align*}
\sup_{\boldx \in \R^{2d}}  \P \Bl\{ {\blam}_{\pas; L} \text{ is } \left( \theta,\,E \right)\text{-nonsuitable}\Br\}  \leq \tfrac{1}{2} L^{-p} + \tfrac 1 2 ( 6Y + 2  ) ^{2d (J+1)} p_{\ell}^{J+1}.
\end{align*}
That is to say that for large enough $\ell,$
\begin{align}
2 \sup_{\boldx \in \R^{2d}}  \P \Bl\{\boldlambda_{\pas; L}(\boldx) \text{ is } \left( \theta,\,E \right)\text{-nonsuitable}\Br\}    \leq L^{-p} +
 \pa{ ( 6Y + 2  ) ^{2d }p_\ell  }^{J+1}.
\end{align}
This implies that for sufficiently large enough $L_0$, setting $p_k = p_{L_k}$, we have that for every $k \in \N,$
\beq \label{part1mainthmeq1}
2p_{k+1} \leq (L_{k+1})^{-p} +   \left( ( 6Y + 2  ) ^{2d } p_k \right)^{J+1}.
\eeq
If $K_0 = 0,$ then we are done. If not, then we have that $(L_0)^{-p} < p_0$ and proceed to checking whether $K_0 = 1.$ If $K_0 = 1,$ then we are done. If not, then we must have that $(L_1)^{-p} < p_1,$ and by equation \eqref{part1mainthmeq1} we must have
$p_1 < \pa{ ( 6Y + 2  ) ^{2d }   p_0 }^{J+1}.$

We now proceed to check whether $K_0 = 2.$ If $K_0 = 2,$ then we are done. If not, then we must have that $(L_2)^{-p} < p_2,$ and by equation \eqref{part1mainthmeq1} we must have
\begin{align} \label{part1mainthmeq2}
p_2 < \pa{( 6Y + 2  ) ^{2d } p_1}^{J+1} &<  \pa{( 6Y + 2  ) ^{2d }     \pa{( 6Y + 2  ) ^{2d }  p_0}^{J+1}}^{J+1} \notag \\
& = \pa{ ( 6Y + 2  ) ^{  \pa{ 2d + 2d(J+1) } (J+1)}   } p_0^{(J+1)^2}
\end{align}
Thus $(L_2)^{-p} < \pa{( 6Y + 2  ) ^{2d }}^{ (J+1)^2 + (J+1)} p_0^{(J+1)^2}.$

We will carry this process out until we can find such a $K_0$. We will proceed by contradiction and assume
such a $K_0$ does not exists. It is clear that under such assumption, we get that for every $k \in \N,$
\begin{align}   \label{part1mainthmeq3}
L_k^{-p} &< \pa{( 6Y + 2  ) ^{2d }}^{(J+1)+(J+1)^2 +...+(J+1)^k } p_0^{(J+1)^k}   \notag  \\
& =  \pa{( 6Y + 2  ) ^{2d }}^{ (J+1) \pa{1 + (J+1) + \dots + (J+1)^{k-1}} } p_0^{(J+1)^k} \\
& = \pa{( 6Y + 2  ) ^{2d (J+1)}} ^{\tfrac{(J+1)^k - 1}{(J+1) - 1}} p_0^{(J+1)^k}. \notag
\end{align}
Up to now, we only need $J \ge 1$ due to equation \eqref{Jnecessary1}. Thus let us take $J = 1$ and by rewriting equation \eqref{part1mainthmeq3}, we get $L_k^{-p} < \pa{( 6Y + 2  ) ^{4d }} ^{ 2^{k} - 1 } p_0^{2^k}$  for every $k \in \N$, i.e.
\begin{align}
\pa{ ( 6Y + 2  ) ^{4d } } L_{0}^{-p}(Y^{-p} )^{k} &<
\pa{( 6Y + 2  ) ^{4d } p_0}^{2^k} \quad \text{for every } k\in\N.
\end{align}
However, this cannot be true since $( 6Y + 2  ) ^{4d } p_0 < 1$,  and we have reached a contradiction.
\end{proof}

\subsubsection{The second Multiscale Analysis for symmetrized two-particle boxes}

\begin{proposition} \label{2ndMSA}
Let $E\leq E^{(2)} $,  $p > 0$, $0<m_0<{ m^*_\tau}$, $1<\gamma<1+ \frac{p}{p+4d}$. If for some $L_0$ sufficiently large we have
\beq
\sup_{\boldx \in \R^{2d}} \P \Bigl\{ \mathbf{\Lambda}_{\paS; L_0}(\boldx) \, \,  \text{is}\,\, (m_0, \,E)  \text{-nonregular} \Bigr\} \leq L_0 ^{-p},
\eeq
then, setting $L_{k+1} = L_k ^{\gamma}$ for $k = 0, 1, \dots$, we get
\beq
\sup_{\boldx \in \R^{2d}} \P \Bigl\{ \mathbf{\Lambda}_{\paS; L_k} (\boldx)    \,\,\text{is}\,\,(\tfrac{m_0}{2}, \,E)  \text{-nonregular} \Bigr\} \leq L_k^{-p},
\eeq
for all  $ k = 0, 1, \dots $.
\end{proposition}

The proof relies on the following lemma, in the same way \cite[Proposition 3.4]{KlN1} relies on \cite[Lemma 3.5]{KlN1}.

\begin{lemma} \label{deterministicMSA2}
Let $E\in \R$, $L  =  \ell^{\gamma}$, $ p>0$, $1 < \gamma  < \tfrac{2p+4d}{p +4d}$, $0 < m_0 < { m^{*}_\tau}$,   $J \in \N$, and
\beq
\ml \in [\ell^{-\kappa} , m_0]\qtx{where} 0 < \kappa < \min \set{ \gamma -1, \gamma \pa{1 - \beta}, 1}
\eeq
 Given a symmetrized two-particle box, ${\blam} := {\blam}_{{S}; L}(\x)$ with the usual $\ell$-suitable partial cover, suppose
\begin{enumerate}[i)]
\item ${\blam}$ is $(E,\beta)$-nonresonant.
\item There exist at most $J$ pairwise $\ell$-distant symmetrized two-particle boxes in the suitable partial cover $\cC_{L, \ell}(\x)$ that are $(m_\ell, E)$-nonregular.
\item Every symmetrized two-particle box with center belonging to  $\Xi_{L, \ell}(\x)$ whose length is $j\pa{8\ell+1}$ with $j \in \set{1, \dots, J}$ is $(E,\beta)$-nonresonant.
\end{enumerate}
Then, for $L$ large enough, ${\blam}$ is $(m_L, E)$-regular, with
\begin{align} \label{mlmL}
\ml > m_ L  \ge   \ml           -  \tfrac{1}{2 \ell^{\kappa}} \ge \tfrac{1}{L^{\kappa}} .
\end{align}
\end{lemma}

The proof of Lemma \ref{deterministicMSA2} is the same as Lemma \ref{deterministicMSA1}, with the corresponding changes to \emph{regular} boxes instead of \emph{suitable} ones. The proof of Proposition \ref{2ndMSA} relies on Lemma \ref{deterministicMSA2} in the same way the proof \cite[Proposition 3.4]{KlN1} uses \cite[Lemma 3.5]{KlN1}.

\subsubsection{The third Multiscale Analysis for symmetrized two-particle boxes}

\begin{proposition} \label{3dMSA}
Given $0 < \zeta_1 < \zeta_0 < 1 $ , $Y = \max \set{34^{\tfrac{1}{1 - \zeta_0}} , 4^{\tfrac{1}{\zeta_0}} }$, and $E\leq E^{(2)} $. If for some $L_0 \in \N$ sufficiently large we have
\beq
\sup_{\boldx \in \R^{2d}} \P \Bigl\{ \mathbf{\Lambda}_{\paS; L_0}(\boldx) \, \,  \text{is}\,\,(\zeta_0, \,E)  \text{-nonSES} \Bigr\} \leq (6Y+2)^{-4d},
\eeq
then, setting $L_{k+1} = Y\,L_k,$ for $k = 0, 1, 2, ...,$ there exists $K_0 \in \N$ such that for every $k \geq K_0$ we have
\beq
\sup_{\boldx \in \R^{2d}} \P \Bigl\{ \mathbf{\Lambda}_{\paS;L_k} (\boldx)    \,\,\text{is}\,\,(\zeta_0, \,E)  \text{-nonSES} \Bigr\} \leq e^{-L_k^{\zeta_1}}.
\eeq
\end{proposition}

The proof of Proposition \ref{3dMSA} relies on the following lemma, in the same way as \cite[Proposition 3.6]{KlN1} is based on \cite[Lemma 3.7]{KlN1}.

\begin{lemma} \label{deterministicMSA3}
Let $0 < \zeta_1 < \zeta_0 < 1 $ , $Y = \max \set{34^{\tfrac{1}{1 - \zeta_0}} , 4^{\tfrac{1}{\zeta_0}} }$,  $L  = Y \ell$,   $J \in \N$, and $E \in \R$. Given a symmetrized two-particle box, ${\blam} := {\blam}_{{S}; L}(\x)$, with the usual $\ell$-suitable partial cover. Suppose
\begin{enumerate}[i)]
\item ${\blam}$ is $(E,\beta)$-nonresonant.
\item There exists at most $J$ pairwise $\ell$-distant symmetrized two-particle box in the suitable partial cover $\cC_{L, \ell}(\x)$ that are $(\zeta_0, E)$-nonSES.
\item Every symmetrized two-particle box with center belonging to  $\Xi_{L, \ell}(\x)$ whose length is $j\pa{8\ell+1}$ with $j \in \set{1, \dots, J}$ is $(E,\beta)$-nonresonant.
\end{enumerate}
Then ${\blam}$ is $(\zeta_0, E)$-SES for $L$ large enough.
\end{lemma}

Lemma \ref{deterministicMSA3} can be proved in the same way as Lemma \ref{deterministicMSA1} and Lemma \ref{deterministicMSA2} adapted to nonSES boxes. We refer the reader to \cite[Lemma 3.7]{KlN1} for details.

\subsubsection{The fourth Multiscale Analysis for symetrized two-particle boxes }

 In this section we proceed with the energy-interval Multiscale Analysis. We fix $\zeta,\tau,\beta,\zeta_1,\zeta_2,\gamma$ as in \eqref{constant} and take $L=\ell^\gamma$.

\begin{definition}
Let $\boldlambda_{ \pas ; L}(\bold{x})$ be a non-interactive symmetrized two-particle box with the usual $\ell$ suitable partial cover, and consider an energy $E \leq E^{(2)}$. Then:

\begin{enumerate}
\item $\boldlambda_{ \pas ; L}(\bold{x})$ is not $E$-Lregular \emph(for left regular\emph) if and only if there are two partially separated  boxes in $\cC_{L, \ell}(x_{1})$
that are $(m_{\tau}^*, \, E-\mu)$-nonregular
for some $\mu \in \sigma \pa{ H_{ \Lambda_{L}({x}_{ 2} )  }  }\cap (-\infty,E^{(1)}]$.

\item $\boldlambda_{ \pas ; L}(\bold{x})$ is not $E$-Rregular \emph(for right regular\emph) if and only if there are two partially separated  boxes in $\cC_{L, \ell} (x_{2})$   that are $(m_{\tau}^*, \, E-\lambda)$-nonregular for some $\lambda \in \sigma \pa{ H_{ \Lambda_{L}({x}_{1} )  }   } \cap (-\infty,E^{(1)}]$.

\item $\boldlambda_{ \pas ; L}(\bold{x})$ is  $E$-preregular if and only if  $\boldlambda_{ \pas ; L}(\bold{x})$ is  $E$-Lregular and $E$-Rregular.

\end{enumerate}
\end{definition}

The following Lemma can be proven as in \cite[Lemma 3.15]{KlN1},
\begin{lemma} \label{prereg}
Let $E_0 \leq E^{(2)}$, $I = [E_0-\delta_\tau, \, E_0+\delta_\tau]\subset (-\infty,E^{(1)}]$,  and consider a non-interactive symmetrized two-particle box $ \boldlambda_{ \pas ; L}(\bold{x})$. Then
\begin{enumerate}
\item $\P \set{\boldlambda_{ \pas ; L}(\bold{x}) \text{ is not }E\text{-Lregular for some } E \in I} \leq  L^{3d} e^{-\ell^{\tau}}$,
\item $\P \set{\boldlambda_{ \pas ; L}(\bold{x}) \text{ is not }E\text{-Rregular for some } E \in I}  \leq  L^{3d} e^{-\ell^{\tau}}$.

\end{enumerate}
We conclude that
   \begin{align}
& \P \bigl\{ \boldlambda_{ \pas ; L}(\bold{x}) \text{ is not }E\text{-preregular for some } E \in I  \bigr\} \leq  2L^{3d} e^{-\ell^{\tau}} .
\end{align}

\end{lemma}

\begin{definition} \label{CNR}
Let ${\blam} =   \boldlambda_{\pas, L}(\x)$ be a non-interactive two-particle box, and consider an energy $E \leq E^{(2)}$. Then:

\begin{enumerate}

\item $\boldlambda_{\pas; L}(\x)$ is $E$-left nonresonant \emph(or LNR\emph) if and only if every box
$\Lambda_{9 \ell}( a) \subseteq \Lambda_{L}(x_{1})$, with $ a \in \suitc (x_{1})$ , is $(E-\mu,\beta)$-nonreso\-nant for every
$\mu \in \sigma \pa{H_{\Lambda_{L}(x_{2})}}\cap (-\infty,E^{(1)}]$. Otherwise we say $\boldlambda_{\pas; L}(\x)$ is $E$-left resonant \emph(or LR\emph).

\item $\boldlambda_{\pas; L}(\x)$ is $E$-right nonresonant \emph(or RNR\emph) if and only if every box $\Lambda_{9 \ell}( a) \subseteq \Lambda_{L}(x_{2})$, with $ a \in \suitc (x_{2})$, is $(E-\lambda,\beta)$-nonreso\-nant for every
$\lambda \in \sigma \pa{H_{\Lambda_{L}(x_{1})}}\cap (-\infty,E^{(1)}]$. Otherwise we say $\boldlambda_{\pas; L}(\x)$ is $E$-right resonant \emph(or RR\emph).
\item We say $\boldlambda_{\pas; L}(\x)$ is
$E$-highly nonresonant \emph(or HNR\emph) if and only if ${\blam}$ is $E$-nonresonant, that is,
$E$-LNR and $E$-RNR.
\end{enumerate}

\end{definition}

\begin{lemma}\label{part2firstthm}
Let $E \in \R,$ and
$  \boldlambda_{\pas; L}(\x)$
be a non-interactive two-particle box. Assume that the following are true:
\begin{enumerate}
\item $\boldlambda_{\pas; L}(\x)$ is $E$-HNR.
\item $\boldlambda_{\pas; L}(\x)$ is $E$-preregular.
\end{enumerate}
Then $\boldlambda_{\pas; L}(\x)$ is $\pa{ m(L), E}$-regular, where
 \beq   \label{mL}
m(L) = m_{\tau}^* - \frac{1}{2L^\kappa}- \frac{6(d+1)\log 2L}{L},
 \eeq
 where $\kappa$ is defined in Lemma \ref{deterministicMSA2}.

\end{lemma}
The proof of the statement for $\blam_{L}(\x)$ and  $\blam_{L}(\pi(\x))$ follows the arguments in \cite[Lemma 5.17]{KlN2}. The desired result for $\boldlambda_{\pas; L}(\x)$ is a consequence of  \eqref{PIequality}.

\begin{lemma} \label{LRresonant}
Let $E \leq E^{(2)}$ and $ \boldlambda_{\pas; L}(\x) $ be a non-interactive symmetrized two-particle box.

\begin{enumerate}
\item If $\boldlambda_{ \pas; L}(\x) $ is  $E$-right resonant, then there exists a two-particle rectangle
\beq
{\blam}_1 = \Lambda_{L}(x_1) \times  \Lambda_{9\ell} (u) ,
\eeq
where $u \in \suitc (x_2)$, and $  \boldlambda_{9\ell} (u)   \subseteq \Lambda_{L}(x_2) $,  such that
\beq
\dist \pa{ \sigma \pa{ H_{{\blam}_1} } , E    } < \tfrac{1}{2} e^{-\pa{9\ell}^{\beta}} \leq \tfrac{1}{2}  e^{-\ell^{\beta}}.
\eeq
Therefore,
\beq
\dist \bigl( \sigma \left(   H_{{\blam}} \bigr) , E    \right) =\dist \bigl( \sigma \left(   H_{{\blam}_1} \bigr) , E    \right) < \tfrac{1}{2} e^{-\pa{9\ell}^{\beta}} \leq \tfrac{1}{2}  e^{-\ell^{\beta}}.
\eeq

\item If $\boldlambda_{ \pas; L}(\x)$ is E-left resonant, then there exists a two-particle rectangle
\beq
{\blam}_1 = \Lambda_{9\ell} (u)  \times \boldlambda_{L}(x_2) ,
\eeq
where   $u \in \suitc (x_1)$, and $  \boldlambda_{9\ell}(u)   \subseteq \boldlambda_{L}(x_1) $,   such that
\beq
\dist \bigl( \sigma \left(   H_{{\blam}_1} \bigr) , E    \right) < \tfrac{1}{2} e^{-\pa{9 \ell}^{\beta}} \leq \tfrac{1}{2} e^{-\ell^{\beta}}.
\eeq
Therefore,
\beq
\dist \pa{ \sigma \pa{H_{{\blam}}} ,E} = \dist \pa{ \sigma \pa{H_{{\blam}_1}}, E  } < \tfrac{1}{2} e^{-\pa{9\ell}^{\beta}} \leq \tfrac{1}{2}  e^{-\ell^{\beta}}.
\eeq
\end{enumerate}
\end{lemma}
The proof follows the arguments in the proof of \cite[Lemma 5.18]{KlN2}, plus \eqref{PIequality}.

We now state the energy interval multiscale analysis.
Given $m > 0$,  $L \in \N$, $\x, \, \y \in \R^{2d}$, and an interval $I$, we define the event
\begin{align}\label{defRxy}
&\hskip -18pt R \left( m, \, I,\, \x, \, \y, \, L  \right) =    \\
& \hskip -10pt  \set{ \exists \, E \in I \sqtx{such that}\boldlambda _{\pas; L} (\boldx) \text{ and } \boldlambda _{\pas; L} (\boldy)  \text{ are not } \left(m, E\right)\text{-regular}  } . \notag
\end{align}

\begin{proposition} \label{part4mainthm}
 Let $\zeta, \, \tau, \beta, \, \zeta_1,\, \zeta_2,\, \gamma$ as in \eq{constant}.  Given $E\leq E^{(2)}$, there exists a length scale $\mathcal Z$ such that if for some $L_0 \geq \mathcal Z$ we can verify
\beq
\sup_{\boldx \in \R^{2d}} \P \Bigl\{ \mathbf{\Lambda}_{\paS;L_0} (\boldx)    \,\,\text{is}\,\,( 2L_0^{\zeta_0-1}, \,E)  \text{-nonregular} \Bigr\} \leq e^{-L_0^{\zeta_1}},
\eeq
with   $m_0:=(2L_0^{\zeta_0-1}-6\log 2 L_0^{-1}) <m^*_{\tau}$,
then, there exists $\delta=\delta(L_0,\zeta_0,\beta)$ such that, setting $I=[E-\delta,E+\delta]\subset (-\infty,E^{(2)}]$ and $L_{k+1} = L_k^{\gamma} = L_0^{\gamma^k}$ for  $k = 0, 1, 2, ...$, we have
\beq
\P \Bigl\{R \left( \frac{m_0}{2}, \, I,\, \x, \, \y, \, L_k \right) \Bigr\} \leq e^{-L_k^{\zeta_2}}
\eeq
for   every pair of partially separated two-particle symmetrized boxes $\boldlambda_{\pas; L_k} (\boldx)$ and $\boldlambda_{\pas; L_k} (\boldy)$.
\end{proposition}

\begin{proof}
We can proceed as in \cite[Proposition 3.13]{KlN1} to deduce from the hypothesis that, setting
\[ \delta=\frac{1}{2}e^{-L_0^{\zeta_0}-2L_0^\beta},\]
we have
\beq \sup_{x\in\R^{2d}}\P \Bigl\{ \exists E\in I\, \mbox{such that }\, \blam_{S;L_0}(\x)\mbox{ is }(m_0,E)\mbox{-nonregular} \Bigr\}\leq e^{-L_0^{\zeta_1}}.
\eeq
Then, we can argue as in \cite[Eq. 5.37 to 5.38]{GK1} and obtain
\beq
\P \Bigl\{R \left( m_0, \, I,\, \x, \, \y, \, L_0 \right)\Bigr\}  \leq e^{-L_0^{\zeta_2}}\eeq
for every pair of partially separated two-particle symmetrized boxes $\boldlambda_{\pas; L_0} (\x)$ and $\boldlambda_{\pas; L_0} (\y)$.

Given $\ell$ (sufficiently large) and $0<m_{\ell} < m_{\tau}$, we set $L = \ell^{\gamma}$ and take  $m_L$ as in \eq{mlmL}.  If $\ell$ is large, we have $m(\ell) > m_\ell$, where  $m(\ell)$ is given in  \eqref{mL}, and conclude that $m(L) \ge m(\ell) > m_\ell >m_L$.

We start by showing that if
\beq
\P \Bigl\{R \left( m_{\ell}, \, I,\, \x, \, \y, \, \ell  \right) \Bigr\}\leq e^{-\ell^{\zeta_2}}
\eeq
for every pair of partially separated two-particle boxes $\boldlambda_{\pas;\ell}(\x)$ and $\boldlambda_{\pas;\ell}(\y)$,
 then
\beq
\P \Bigl\{R \left( m_{L}, \, I,\, \x, \, \y, \, L  \right) \Bigr\}\leq e^{-L^{\zeta_2}}
\eeq
for every pair of partially {separated} symmetrized two-particle boxes $\boldlambda_{\pas; L}(\x)$ and $\boldlambda_{\pas; L} (\y)$.

Let $\boldlambda_{\pas; L} (\x)$ and  $\boldlambda_{\pas; L} (\y)$ be a pair of partially separated symmetrized two-particle boxes.  Let $J \in 2\N$. Let $\cB_{J}$ be the  the event that there exists $E \in I$ such that  either ${\mathcal C}_{L,\ell} (\boldx)$ or ${\mathcal C}_{L,\ell} (\boldy)$ contains $J$ pairwise $\ell$-distant interactive symmetrized boxes that are $(m_{\ell}, \, E)$-nonregular, and let
$\cA$ be  the event that there exists $E \in I$ such that
either ${\mathcal C}_{L,\ell} (\boldx)$ or ${\mathcal C}_{L,\ell} (\boldy)$ contains one non-interactive symmetrized box   that is not $E$-preregular.
If $\bom \in \cB_{J}^{c} \cap \cA^{c}$, then for all $E \in I$ the following holds:
\begin{enumerate}
\item ${\mathcal C}_{L,\ell} (\boldx)$ and ${\mathcal C}_{L,\ell} (\boldy)$ contain at most $J-1$ pairwise $\ell$-distant interactive $(m_{\ell}, \, E)$-nonregu\-lar boxes.

\item Every interactive box in ${\mathcal C}_{L,\ell} (\boldx)$ and ${\mathcal C}_{L,\ell} (\boldy)$ is $E$-preregular.
\end{enumerate}

We also define the event
\beq
\cU_{J} = \bigcup_{\boldlambda^{\pr} \in \cM_{\boldx},\boldlambda^{\pr\pr} \in  \cM_{\boldy}} \set{ \dist \left( \sigma(H_{\boldlambda^{\pr}}), \sigma(H_{\boldlambda^{\pr\pr}}) \right) < e^{-\ell^{\beta}} },
\eeq
where,
given a symmetrized two-particle box ${\blam}_{\pas, L} (\bolda)$,   by $\cM_{\bolda}$ we denote  the collection of all symmetrized two-particle rectangles of the following three types:

\begin{enumerate}
\item ${\blam}_{\pas; L} (\bolda)$,
\item $\boldlambda_{\pas; 9j \ell}(\boldu) \subseteq {\blam}_{\pas, L} (\bolda)$, where $\boldu \in \suitc(\bolda)$,  and $j \in \set{1,\,2,\ldots, J }$  ,

\item $\boldlambda$ is a symmetrized two-particle rectangle generated by a two-particle rectangle either of the form $ \Lambda_{9\ell} ( v) \times \Lambda_{L}(a_2)$, where $ v \in \suitc(a_1)$, or $ \Lambda_{L} ( a_1) \times \Lambda_{9\ell}(v)$, where $ v \in \suitc(a_2)$ .
\end{enumerate}
It is clear that if $J \ge 3$, then $\abs{\cM_{\bolda}} <  J \pa{\tfrac{6L}{\ell}  +2 }^{2d} + 2    \pa{\tfrac{6L}{\ell} +2 } ^{d} + 1 < (J+1) \pa{\tfrac{6L}{\ell} +2 }^{2d}$, and hence it follows from Proposition~\ref{Wegner2} that
\beq\label{probUJ}
\P \set{\cU_{\cJ}} \leq  \pa{ (J+1) \pa{\tfrac{6L}{\ell} +2 }^{2d} } ^2  \pa{   16 \norm{\rho}_{\infty} L^{4d}      e^{-\ell^{\beta}}    }.
\eeq
Note that  for $\bom \in \ \cU_J^{c}$ and  $E\in I$, either every symmetrized two-particle rectangle in $\cM_{\boldx}$ is $(E,\beta)$-nonresonant or every symmetrized two-particle rectangle in $\cM_{\boldy}$ is $(E,\beta)$-nonresonant.

Let $\bom \in \cB_{J}^{c} \cap \cA^{c}  \cap \cU_J^{c}$ and $E \in I$.
If every symmetrized two-particle rectangle in $\cM_{\boldx} $  is $(E,\beta)$-nonresonant, then, in particular, every non-interactive symmetrized rectangle in ${\mathcal C}_{L,\ell} (\boldx)$ is $E$-HNR and $E$-preregular, and hence $\pa{m(\ell), \, E}$-regular by Lemma \ref{part2firstthm}.   As $m(\ell)>m_\ell$,
we conclude that  every non-interactive symmetrized box in ${\mathcal C}_{L,\ell} (\boldx)$ is $\pa{m_\ell ,\, E}$-regular. Since  $\bom \in \cB_{J}^{c} \cap \cA^{c}$, ${\mathcal C}_{L,\ell} (\boldx)$ contains at most $J-1$ pairwise $\ell$-distant interactive  $(m_\ell, \, E)$-nonregular boxes in ${\mathcal C}_{L,\ell} (\boldx)$, and all other symmetrized boxes in ${\mathcal C}_{L,\ell} (\boldx)$ are  $(m_\ell, \, E)$-regular, it follows from Lemma~\ref{deterministicMSA2} that   ${\blam}_{\pas; L} (\boldx)$ is $(m_L, \, E)$-regular. If there exists a symmetrized two-particle rectangle in $\cM_{\x}$ that is $(E,\beta)$-nonresonant, then every symmetrized two-particle rectangle in $\cM_{\y}$ must be $(E,\beta)$-nonresonant, and thus ${\blam}_{\pas; L} (\boldy)$ is $(m_L, \, E)$-regular
using the same argument as before.   Thus  for every $E \in I$ either ${\blam}_{\pas; L} (\boldx)$ is $(m_L, \, E)$-regular or
${\blam}_{\pas; L} (\boldy)$ is $(m_L, \, E)$-regular.    It follows that
\beq
R \pa{m_{L}, \, I,\, \x, \, \y, \, L } \subseteq \pa{\cB_{J}^{c} \cap \cA^{c}  \cap \cU_J^{c}}^c,
\eeq
so
\beq \label{probR}
\P \set{R \pa{ m_{L}, \, I,\, \x, \, \y, \, L} } \leq \P \pa{\cB_{J}} + \P \pa{\cA} + \P \pa{\cU_{J}}.
\eeq

Using independence and Lemma \ref{prereg},
  we get
  \beq\label{probAB}
  \P\set{\cB_{J}}  \leq 2\pa{ \tfrac {6L}{\ell} + 2}^{4d} e^{- \tfrac{J}{2} L^{\tfrac{\zeta_2}{\gamma}}} \qtx{and} \P\set{\cA}  \leq 2\pa{ \tfrac {2L}\ell}^{2Nd}  e^{-L^{\tfrac{\tau}{\gamma}}}.
  \eeq
We now fix
$$
J \in \Bigl( 2 L^{\beta - \tfrac{\zeta_2}{\gamma}},2 L^{\beta - \tfrac{\zeta_2}{\gamma}}+2\Bigr] \cap 2\N, $$
so, for $L$ sufficiently large,  $\P\set{\cB_{J}}  \leq  \tfrac{1}{3} e^{-L^{\zeta_2}}$ , $\P\set{\cA}  \leq \tfrac{1}{3} e^{-L^{\zeta_2}}$, and $\P\set{\cU_{J}}  \leq  \tfrac{1}{3} e^{-L^{\zeta_2}}$, and we conclude from
\eq{probR} that
\beq
\P\set{R \left( m_{L}, \, I,\, \x, \, \y, \, L\right)} \leq e^{-L^{\zeta_2}} .
\eeq

We now take $L_0$ large enough so that  $m(L_0)>m_{L_0}=m_0$  and the above procedure can be carried out with $\ell=L_0$,  $L_{k+1} = L_k^{\gamma}$ for $k=0,1,\ldots$, and $m_{k}\geq m_{k-1}-(2L_{k-1}^{-\kappa})$, where we write $m_k:=m_{L_k}$. To finish the proof, we just need to make sure $m_{k} > \tfrac{m_0}{2}$ for all $k=0,1,\ldots$, which can be done  taking $L_0$ sufficiently large, using the fact $m_0\geq L_0^{-\kappa}$ as in \cite[Eq. 3.46]{KlN1}.
\end{proof}

\subsubsection{Completing the proof of the Bootstrap Multiscale Analysis for symmetrized two-particle boxes}
The proof of Theorem \ref{MSAthm} follows from Propositions \ref{1stMSA}, \ref{2ndMSA}, \ref{3dMSA}, \ref{part4mainthm} as in \cite[Section 6]{GK1}, see also \cite[Section 3.5]{KlN1}. The result holds for all sufficiently large scales by the argument in \cite[Theorem~3.21]{KlN1}.

\section{Extracting dynamical localization from the Bootstrap Multiscale Analysis}\label{S:dynloc}

In this section we prove  Theorem~\ref{dynloc}. To do so,  we present an improvement of \cite[Corollary 1.7]{KlN1} (stated for  the usual boxes), which we state using the Hausdorff distance $\dist_H$ in the $n$-particle setting. This result holds in the symmetrized distance if the conclusions of the Multiscale Analysis hold with respect to the $\dist_S$.

  Following the arguments in \cite[Section 4.1]{KlN1}, we can prove that $H_{\bom}=H_{\bom}^{(n)}$ exhibits Anderson localization in the interval $I$, that is, for almost every $\bom$, $\sigma \pa{H_{\bom}} \; \cap \; I = \sigma_{pp} \pa{H_{\bom}} \; \cap \; I$ and $\sigma_{cont} \pa{H_{\bom}} \; \cap \; I = \emptyset$. We fix $\nu=(nd+1)/2$ and for $\bolda\in\Z^{nd}$
  let $T_\bolda$ denote the operator on $\ell^2(\Z^{nd})$ given by multiplication by $(1+\norm{\boldx-\bolda}^2)^{\nu/2}$. By $\chi_\x$ we denote the orthogonal projection onto $\delta_\x$. We will work with the following definition, from \cite[Section 3]{GKjsp},

\begin{definition}\label{defn:ZW}
 Given $\bom,\lambda\in\R$ and $\bolda\in\Z^{nd}$, define
\begin{equation}
W_{\lambda,{\bom}}(\bolda):=
\begin{cases}
\displaystyle\sup_{\phi\in\mathcal T_{\lambda,{\bom}}}\frac{ \norm{\chi_{\bolda} \chi_{\{\lambda\}}(H_{\bom})\phi }  }{\norm{T_{\bolda}^{-1} \chi_{\{\lambda\}}(H_{\bom}) \phi} }, & \mbox{if } \chi_{\{\lambda\}}(H_{\bom})\neq 0\\
0, & \mbox{otherwise},
\end{cases}
\end{equation}
where $\mathcal T_{\lambda,{\bom}}=\{\phi \in \ell^2(\Z^{nd});\, \chi_{\{\lambda\}}(H_{\bom})\phi\neq 0\}$, and
\begin{equation}
Z_{\lambda,{\bom}}(\bolda):=
\begin{cases}
\displaystyle\frac{ \norm{\chi_{\bolda} \chi_{\{\lambda\}}(H_{\bom}) }_2  }{\norm{T_{\bolda}^{-1} \chi_{\{\lambda\}}(H_{\bom}) }_2 }, & \mbox{if } \chi_{\{\lambda\}}(H_{\bom})\neq 0\\
0, & \mbox{otherwise}.
\end{cases}
\end{equation}
\end{definition}
We have
\begin{equation}\label{ZW}
Z_{\lambda,{\bom}}(\bolda)\leq W_{\lambda,{\bom}}(\bolda)\leq 1.
\end{equation}

 \begin{theorem}\label{decaykernel0} Let $I$ be  a closed  interval where, either the conclusions of \cite[Theorem~1.5]{KlN1} hold for $H^{(n)}_{\bom}$, or, in the case $n=2$,  the conclusions  of Theorem \ref{MSAthm} hold. Then, for every $\zeta\in(0,1)$ there exists a constant $C_\zeta>0$ such that
 \begin{align}
\E \set{\sup_{\abs{f} \leq 1}  \norm{ \Chi_{\boldx} \, (f\, \Chi_{I})(H_{\bom})   \,\Chi_{\boldy} }_2^2   }&\le \E \set{\sup_{\abs{f} \leq 1}  \norm{ \Chi_{\boldx} \, (f\, \Chi_{I})(H_{\bom})   \,\Chi_{\boldy} }_1   }\\ \notag
& \leq C_{\zeta}  e^{-d_{H}( \boldx , \, \boldy )^{\zeta}} \qtx{for all} \boldx, \boldy \in \R^{nd},
 \end{align}
where the supremum is taken over all bounded Borel functions $f$ on $\R$.
 \end{theorem}

 \begin{proof}
Since $H_{\bom}$ exhibits Anderson localization in $I$ for almost every $\bom$, let $\set{\psi_{j}(\bom)}$ be an orthonormal base of the space  $\text{Ran} \;\Chi_{I}(H_{\bom})$ consisting of eigenfunctions of $H_{\bom} \Chi_{I}(H_{\bom})$ with corresponding eigenvalues $\lambda_{j}(\bom)$. Then
 \begin{align}
 &\hskip 48pt H_{\bom} \Chi_{I}(H_{\bom}) = \sum_{j} \lambda_{j}(\bom) P_{\psi_{j}(\bom)}, \qtx{and} \notag\\
 &\sup_{\abs{f} \leq 1}  \norm{ \Chi_{\boldx} \, (f\, \Chi_{I})(H_{\bom})   \,\Chi_{\boldy} }_1  = \sup_{\abs{f} \leq 1}  \norm{ \Chi_{\boldx} \, \sum_{j} f(\lambda_{j}(\bom)) P_{\psi_{j}(\bom)}   \,\Chi_{\boldy} }_1,
 \end{align}
 where $P_{\psi_{j}(\bom)}$ is the projection onto the space spanned by ${\psi_{j}(\bom)}$.  Since $\abs{f} \le 1$, we have
 \begin{align}
 \norm{ \Chi_{\boldx} \, \sum_{j} f(\lambda_{j}(\bom)) P_{\psi_{j}(\bom)}   \,\Chi_{\boldy} }_1 &  \le \sum_{j} \norm{ \Chi_{\boldx} \, P_{\psi_{j}(\bom)}  \,\Chi_{\boldy} }_1 \notag \\
 &\le  \sum_{j} \norm{ \Chi_{\boldx} \, P_{\psi_{j}(\bom)} }_2 \norm{  P_{\psi_{j}(\bom)} \Chi_{\boldy} }_2.
 \end{align}

 From Definition \ref{defn:ZW} and \eqref{ZW} we get
 \begin{align}
&  \norm{\Chi_{\boldx} \, P_{\psi_{j}(\bom)}  }_2 \le Z_{\lambda_j , \bom} (\x) \norm{T_{\x}^{-1} P_{\psi_{j} (\bom)} }_2 \le W_{\lambda_j, \bom} (\x) \norm{T_{\x}^{-1} P_{\psi_{_j (\bom)}}}_2  \qtx{and} \notag \\
& \hskip 24pt \norm{\Chi_{\boldy} \, P_{\psi_{j} (\bom)} }_2 \leq W_{\lambda_j, \bom} (\y) \norm{T_{\y}^{-1} P_{\psi_{j} (\bom)} }_2.
 \end{align}
Hence,
 \begin{align}
&  \norm{ \Chi_{\boldx} \, \sum_{j} f(\lambda_{j}(\bom)) P_{\psi_{j}(\bom)}   \,\Chi_{\boldy} }_1   \\
& \hskip 24pt \le \sum_{j}  W_{\lambda_j, \bom} (\x) \norm{T_{\x}^{-1} P_{\psi_{j} (\bom)} }_2  W_{\lambda_j, \bom} (\y) \norm{T_{\y}^{-1} P_{\psi_{j} (\bom)} }_2  \notag \\
  & \hskip 36pt  \le    \norm{ W_{\lambda_j, \bom}(\x) W_{\lambda_j, \bom}(\y)}_{L^{\infty} \pa{I, d \mu_{\bom}(\lambda) }} \sum_{j} \norm{T_{\x}^{-1} P_{\psi_{j} (\bom)}   }_2 \norm{T_{\y}^{-1} P_{\psi_{j} (\bom)}  }_2 . \notag
 \end{align}
 By the Cauchy-Schwarz inequality, we get
 \begin{align}
&  \sum_{j} \norm{T_{\x}^{-1} P_{\psi_{j} (\bom)} }_2 \norm{T_{\y}^{-1} P_{\psi_{j} (\bom)}  }_2  \notag \\
& \hskip 36pt \le \pa{  \sum_{j}  \norm{T_{\x}^{-1}   P_{\psi_{j} (\bom)}  }_{2}^{2}  }^{1/2}
\pa{  \sum_{j}   \norm{T_{\y}^{-1}   P_{\psi_{j} (\bom)}  }_{2}^{2}  }^{1/2}.
 \end{align}
 
For all  $\x \in \R^{nd}$ we have
 \begin{align}
 \sum_{j}  \norm{T_{\x}^{-1}    P_{\psi_{j} (\bom)}  }_{2}^{2} & = \sum_{j}  \tr \set{T_{\x}^{-1}    P_{\psi_{j} (\bom)}^2 T_{\x}^{-1}}\\ \notag
 & = \tr\set{ T_{\x}^{-1} \Chi_{I}(H_{\bom}) T_{\x}^{-1}}\le C_{nd}<\infty.
 \end{align}
 (See, e.g.,  \cite[Lemma 4.1]{KlN1} and its proof.)
 Thus we conclude that {for all} $\x , \y  \in \R^{nd}$ and almost every  $\bom$ we have
 \begin{align}\label{decaykernel-est5}
 & \norm{ \Chi_{\boldx} \, \sum_{j} f(\lambda_{j}(\bom)) P_{\psi_{j}(\bom)}   \,\Chi_{\boldy} }_1 \\
& \hskip 48pt \le
  C_{d,n}   \norm{ W_{\lambda_j, \bom}(\x) W_{\lambda_j, \bom}(\y)}_{L^{\infty} \pa{I, d \mu_{\bom}(\lambda) }}. \notag
 \end{align}

 To study $ \norm{ W_{\lambda_j, \bom}(\x) W_{\lambda_j, \bom}(\y)}_{L^{\infty} \pa{I, d \mu_{\bom}(\lambda) }}$, we divide it into two cases. For the first case, consider  $d_{H}(\x, \y) < L_{\zeta}$.
  Note that we always have
 \begin{align} \label{eigenest}
 & \norm{\Chi_{\bolda} \psi} \le \norm{ \Chi_{\bolda} T_{\bolda}} \norm{T_{\bolda}^{-1} \psi} \le \norm{T_{\bolda}^{-1} \psi} \qtx{for every} \bolda \in \R^{nd}.
 \end{align}
Hence,
\begin{align}
 \E \pa{  \norm{W_{\lambda, \bom}(\x) W_{\lambda, \bom}(\y)}_{L^{\infty} \pa{I, d \mu_{\bom} (\lambda) }}  }  & \le  e^{d_{H}(\x, \y)^{\zeta}} e^{-d_{H}(\x, \y)^{\zeta}}  \\
  &\ \le  e^{L_{\zeta}^{\zeta}} e^{-d_{H}(\x, \y)^{\zeta}} = C_{\zeta}  e^{-d_{H}( \boldx , \, \boldy )^{\zeta}}. \notag
  \end{align}
For the case $d_{H}(\x, \y) = L \ge L_{\zeta}$, we consider the two cases: when $\bom \in R\pa{I, m, L, \x, \y}$ and when $\bom \notin R\pa{I, m, L, \x, \y}$ (recall \eq{defRxy}). By Equation \eqref{eigenest}, for every $\bom \in R\pa{I, m, L, \x, \y}$
 \beq
 \abs{W_{\lambda_j, \bom}(\x) W_{\lambda_j, \bom}(\y)} \le 1 \qtx{for } \lambda_j \in I .
 \eeq

 For the case that $\bom \notin R\pa{I, m, L, \x, \y}$ we have that for $E \in I$,  either
 $\blam_{\sharp;L}(\x)$ or $\blam_{\sharp;L}(\y)$  is $(m,E)$-regular, where
 $\sharp=\infty$ if the conclusions of \cite[Theorem~1.5]{KlN1} hold, and  $\sharp=S$ if the conclusions of Theorem \ref{MSAthm} hold.
 Without loss of generality, say
 $\blam_{\sharp;L}(\x)$   is $(m,E)$-regular. Suppose $H \psi = E \psi$ with $\norm{\psi} = 1$, then
 (see \cite[Eq. 4.8]{KlN1})
 \beq
 \norm{\Chi_{\x} \psi} \le e^{-m L/2} \norm{\Chi_{\partial_+ \blam^{\mathrm{c}}_{\sharp;L}(\x)} \psi},
 \eeq
 where $\blam^{\mathrm{c}}_{\sharp;L}(\x)$ denotes the connected component of 
$\blam_{\sharp;L}(\x)$ containing $\x$.  (Note $\blam^{\mathrm{c}}_{\infty;L}(\x)=\blam_{\infty;L}(\x)$.)   

Since
 \beq
 \norm{\Chi_{\partial_+ \blam^{\mathrm{c}}_{\sharp;L}(\x)} \psi}\le \norm{\Chi_{\partial_+ \blam^{\mathrm{c}}_{\sharp;L}(\x)} T_{\x}  }    \norm{  T_{\x}^{-1}\psi},
 \eeq
and, in either case (recall $n=2$ if $\sharp=S$)
\beq
\norm{\Chi_{\partial_+ \blam^{\mathrm{c}}_{\sharp;L}(\bolda)} T_{\bolda}  } \le L^{\nu} \qtx{for every} \bolda \in \R^{nd},
\eeq
we obtain,
\begin{align}
\norm{\Chi_{\x} \psi} \le e^{-m L/2} L^{\nu}  \norm{  T_{\x}^{-1}\psi}.
\end{align}

Thus, for $\bom \notin R\pa{I, m, L, \x, \y}$, we get
\beq
 \abs{W_{\lambda_j, \bom}(\x) W_{\lambda_j, \bom}(\y)} \le  e^{-m L/2} L^{\nu} \qtx{for  } \lambda_j \in I .
 \eeq
 We can decompose the expectation in two parts corresponding to the set of $\omega\in R\pa{I, m, L, \x, \y}$ and its complement, and obtain (recall that if  $\sharp=S$ we have $n=2$, so $d_{H}(\x, \y)=d_{S}(\x, \y)$)
 \begin{align}
 & \E \pa{  \norm{W_{\lambda, \bom}(\x) W_{\lambda, \bom}(\y)}_{L^{\infty} \pa{I, d \mu_{\bom} (\lambda) }}  } \\
 & \hskip 24pt \le   L^{\nu} e^{-mL/2} +  e^{-L^{\zeta}} \le e^{-L^{\zeta^{\pr}}} = e^{-(d_{H}(\x, \y))^{\zeta^{\pr}}}, \notag
\end{align}
provided $ L_{\zeta}$ is large enough and $\zeta^{\pr} < \zeta$.
Hence, we get
\begin{align}
\E \set{\sup_{\abs{f} \leq 1}  \norm{ \Chi_{\boldx} \, (f\, \Chi_{I})(H_{\bom})   \,\Chi_{\boldy} }_2^2   }&\le \E \set{\sup_{\abs{f} \leq 1}  \norm{ \Chi_{\boldx} \, (f\, \Chi_{I})(H_{\bom})   \,\Chi_{\boldy} }_1   }\\ \notag
& \leq C_{\zeta}  e^{-d_{H}( \boldx , \, \boldy )^{\zeta}} \qtx{for all} \boldx, \boldy \in \R^{nd}.
 \end{align}
   \end{proof}

\begin{corollary} Let $I$ be  a closed  interval where, either the conclusions of \cite[Theorem~1.5]{KlN1} hold for $H^{(n)}_{\bom}$, or, in the case $n=2$,  the conclusions  of Theorem \ref{MSAthm} hold. Then
 \begin{align}\label{dyn1}
\E \set{\sup_{\abs{f} \leq 1}  \norm{ \scal{d_H(\bX,\y)}^{\frac p 2} (f\, \Chi_{I})(H_{\bom})   \,\Chi_{\boldy} }_2^2   }<\infty\,  \qtx{for all}  \boldy \in \R^{nd},
 \end{align}
where the supremum is taken over all bounded Borel functions $f$ on $\R$.
\end{corollary}

\begin{proof}
\begin{align}
&\E \set{\sup_{\abs{f} \leq 1}  \norm{ \scal{d_H(\bX,\y)}^{\frac p 2} (f\, \Chi_{I})(H_{\bom})   \,\Chi_{\boldy} }_2^2   } \notag\\
\notag
& \qquad \le C_1 \sum_{\x \in  \Z^{nd}}  \scal{d_H(\x,\y)}^{p}\E \set{ \sup_{\abs{f} \leq 1} \norm{ \Chi_{\x} f \pa{H_{\bom} }  \Chi_{I} \pa{H_{\bom}} \Chi_{\y} }^{2}_{2}}  \\
& \qquad \le C_1 \sum_{\x \in  \Z^{nd}}  \scal{d_H(\x,\y)}^{p}\E \set{ \sup_{\abs{f} \leq 1} \norm{ \Chi_{\x} f \pa{H_{\bom} }  \Chi_{I} \pa{H_{\bom}} \Chi_{\y} } _{1}}  \\
\notag & \qquad \le  C_2\,   \sum_{\x \in  \Z^{nd}}  \scal{d_H(\x,\y)}^{p} e^{-d_{H}( \boldx , \, \boldy )^{\zeta}}.
 \end{align}
 The result follows from the fact $ \sum_{\x \in  \Z^{nd}}  \scal{d_H(\x,\y)}^{p} e^{-d_{H}( \boldx , \, \boldy )^{\zeta}} < \infty$ shown in  \cite[Lemma~A.3]{AWmp}.
\end{proof}


\appendix

\section{Wegner estimates}\label{apW}

\begin{theorem}\label{t:wesharp} Let   $H_{\bom}^{(n)}$ be  the $n$-particle Anderson model.
Consider $\blam \subset \Z^{nd}$ such that:
\begin{enumerate}
\item If $\sharp=\infty,  S$,  $\blam= {\blam}_{\sharp; \bL}^{(n)} (\x)$ is  an $n$-particle $\sharp$-rectangle of center ${\x}=(x_1,\ldots,x_n)\in \R^{nd}$ and  sides $\bL=(L_1,L_2,\ldots,L_n)\in [1,\infty)^n$, take $\Gamma=\L_{L_k}(x_k)$ for some $k\in\{1,...,n\}$, and  set $L=\max \set{L_1,L_2,\ldots,L_n}$.

\item  If $\sharp=H$,  $\blam= {\blam}_{H; L}^{(n)} (\x)\subset \mathbb Z^{nd}$ is  an $n$-particle $H$-box of center ${\x}=(x_1,\ldots,x_n)\in \R^{nd}$ and  side $L\ge 1$, and take $\Gamma= \bigcup_{k=1}^n \L_{L}(x_k)$.

\end{enumerate}

 Then, for any interval $I\subset \mathbb R$ we have
\be\label{wesharp} \E_\Gamma\pa{{\mbox{\rm tr}}\chi_I(H_{\omega})} \leq C\up{\sharp}_n\norm{\rho}_\infty L^{nd},\ee
where  $C\up{\infty}_n=n$, $C\up{S}_n=n (n!)$, and  $C\up{H}_n=n^{2n+1}$.   In particular, for any $E\in\R$ and $\eps>0$, we have
\be \P_\Gamma\pa{\norm{G_{\blam}(E)} \geq \tfrac{1}{\eps}} =\P_\Gamma  \Bl\{d\,(\sigma(H_{\boldlambda}), E) \leq {\eps} \Br\}\leq 2C_n^{(\sharp)}\norm{\rho}_\infty \eps L^{nd}.
\ee
\end{theorem}
 For  $\sharp=\infty$ this is  \cite[Theorem 2.3]{KlN1}.  This proof can be modified to give the result for  $\sharp=S,H$. The appearance of  the constant $C\up{\sharp}_n$  in \eq{wesharp} is due to the fact that  in the proof of  \cite[Theorem 2.3]{KlN1} we need to take into account the geometry of the boxes $ {\blam}_{\sharp; \bL}^{(n)} (\x)$.

\section{Combes-Thomas estimate for restrictions of discrete Schr\"odinger operators to arbitrary subsets}\label{apCT}

For convenience we state and prove a  Comes-Thomas estimate  for restrictions of discrete Schr\"odinger operators to arbitrary subsets (cf.  \cite{GKdecay,Ki}).

\begin{theorem}\label{thmCT}
Let $H = -\Delta + V$ be a discrete Schr\"odinger operator on $\ell^{2}(\Z^d)$, where $\Delta$ is the  centered Laplacian operator. Given $S \subset \Z^d$,  let $H_S$  be the  restriction of $\Chi_S H \Chi_S$ to $\ell^2(S)$.  Then for every $z \notin \sigma(H_S)$,  setting $\eta_z = \dist \pa{z, \sigma(H)}$, for all $\eps \in ( 0, 1)$ we get
\beq\label{CTest}
\abs{ \scal{\delta_x, (H_S - z)^{-1} \delta_y}  }   \le \tfrac{1}{\eta_z (1-\eps)} e^{- \log \pa{ \tfrac{\eps \eta_z}{2d} + 1 } \norm{y - x}}\qtx{for all} x,y \in S.
\eeq
\end{theorem}

\begin{proof}
Given  $v = (v_1, \dots, v_d) \in \R^d$, let $M_v$ be the multiplication operator given by the function $e^{v \cdot x}$ and $U_v$ be the multiplication operator given by the function $e^{- v \cdot x}$ on $\ell^{2}(\Z^d)$. We set
\beq
H_{v} = M_v H_S U_v  = M_v (-\Delta_S + V_S) U_v = M_v (-\Delta_S) U_v + V_S.
\eeq

Let $\set{e_1, \dots, e_d}$ be the standard basis for $\R^d$, and set
\beq
e_{d+j} = -e_j \qtx{for} j = 1, \dots, d.
\eeq
Given  $\psi \in \ell^{2}(\Z^d)$, we have
\begin{align}
&\pa{-\Delta_S U_v \psi} (x) = -\sum_{y \in S; \, \abs{y-x}_1 = 1} e^{-v \cdot y} \psi(y) \notag\\
&\quad  = - \sum_{j = 1, \dots, 2d; \, x+e_j \in S} e^{-v \cdot (x+e_j)} \psi(x+e_j)  = - \sum_{j = 1, \dots, 2d; \, x+e_j \in S} e^{-v \cdot x} e^{-v \cdot e_j} \psi(x+e_j). \notag
\end{align}
Hence,
\begin{align}
&\pa{M_v \pa{-\Delta_S} U_v \psi} (x) =  - \sum_{j = 1, \dots, 2d; \, x+e_j \in S} e^{v \cdot x} e^{-v \cdot x} e^{-v \cdot e_j} \psi(x+e_j) \\
& \quad = - \sum_{j = 1, \dots, 2d; \, x+e_j \in S}  e^{-v \cdot e_j} \psi(x+e_j) = - \sum_{j = 1, \dots, 2d; \, x+e_j \in S}  \pa{ e^{-v \cdot e_j} -1 +1}\psi(x+e_j) \notag \\
&\quad = - \Delta_S \psi (x) - \sum_{j = 1, \dots, 2d; \, x+e_j \in S}  \pa{e^{-v \cdot e_j} -1} \psi(x+e_j). \notag
\end{align}
 Let us define the operator $B_v$ by
\begin{align}
\pa{ B_v \psi} (x) =  \sum_{j = 1, \dots, 2d; \, x+e_j \in S}  \pa{e^{-v \cdot e_j} -1} \psi(x+e_j).
\end{align}
Then
\begin{align}
\abs{\pa{ B_v \psi} (x) } & = \abs{ \sum_{j = 1, \dots, 2d; \, x+e_j \in S}  \pa{e^{-v \cdot e_j} -1} \psi(x+e_j) } \notag \\
& \le  \max_{i= 1, \dots, d} \set{e^{\abs{v_i}} -1 } \sum_{j = 1, \dots, 2d; \, x+e_j \in S}  \abs{ \psi(x+e_j) }\\ \notag &
=    (e^{\norm{v}} -1) \pa{\Delta_S \abs{ \psi}}(x) , \notag
\end{align}
so
\beq
\norm{B_v} \le \norm{\Delta_S}  (e^{\norm{v}} -1)\le  2d (e^{\norm{v}} -1) .
\eeq

Thus, for $z \notin \sigma(H_S)$, letting  $\eta_z = \dist \pa{z, \sigma(H_S)}$, and requiring
\beq\label{CTcond}
\norm{B_v} \norm{\pa{H_S - z}^{-1}} < 1,
\eeq
we get
\begin{align}
\pa{H_v - z}^{-1} & = \pa{H_S - B_v - z}^{-1}
 = (H_S - z)^{-1} \sum_{k = 0}^{\infty} \pa{B_v (H_S - z)^{-1}}^k ,
 \end{align}
and hence
\begin{align}
\norm{(H_v - z)^{-1}} \le \tfrac{1}{\eta_z} \tfrac{1}{1 - \tfrac{\norm{B_v}}{\eta_z}} = \tfrac{1}{ \eta_z - \norm{B_v}}.
 \end{align}
 If we take $\norm{B_v} \le \eps \eta_z$ for some $\eps < 1$, which can be achieved by requiring $    2d (e^{\norm{v}} -1)   \le \eps \eta_z$, then \eq{CTcond} is satisfied and
 \begin{align}
\norm{(H_v - z)^{-1}} \le   \tfrac{1}{ \eta_z - \norm{B_v}} \le \tfrac{1}{\eta_z (1-\eps)}.
 \end{align}
Moreover, setting $d_{z} = \log \pa{ \tfrac{\eps \eta_z}{2d} + 1 } $, we get
\beq\label{vcondCT}
2d (e^{\norm{v}} -1)   \le  \eps \eta_z  \; \iff \;  \norm{v} \le  \log \pa{ \tfrac{\eps \eta_z}{2d} + 1 } = d_{z}.
\eeq

Thus, for  $v$ satisfying \eq{vcondCT},  we have
\begin{align}
&\abs{ \scal{\delta_x, (H_S - z)^{-1} \delta_y}  }  = \abs{ \scal{ \delta_x, U_v M_v (H_S - z)^{-1} U_v M_v   \delta_y}  }\\ \notag
& \quad = \abs{ \scal{ U_v\delta_x, M_v (H_S - z)^{-1} U_v M_v \delta_y}  }= \abs{ \scal{ e^{-v\cdot x} \delta_x, (H_v - z)^{-1} e^{v \cdot y} \delta_y}  } \\
& \quad = e^{v \cdot \pa{y - x}} \abs{ \scal{  \delta_x, (H_v - z)^{-1} \delta_y}}  \le e^{v \cdot \pa{y - x}} \norm{(H_v - z)^{-1}} \notag \\
& \quad \notag  \le e^{v \cdot \pa{y - x}} \tfrac{1}{\eta_z (1-\eps)} \qtx{for all} x,y \in S.
\end{align}
Choosing  $v =-  d_z \tfrac{y-x}{\norm{y-x}}$, we get
\beq
\abs{ \scal{\delta_x, (H_S - z)^{-1} \delta_y}  } \le\tfrac{1}{\eta_z (1-\eps)}  e^{-d_z \norm{y - x}} = \tfrac{1}{\eta_z (1-\eps)}  e^{- \log \pa{ \tfrac{\eps \eta_z}{2d} + 1 } \norm{y - x}} .
\eeq
\end{proof}

\section{Discrete version of an auxiliary result in \cite{GKduke}}\label{appdiscrete}

 The following is a generalization of \cite[Lemma 6.4]{GKduke} to the discrete setting, stated in Lemma~\ref{GKlem6.4disc}.
\begin{lemma}\label{aGKlem6.4discA}
Let $H_{\bom}^{(n)}$ be a random $n$-particle Schr\"odinger operator satisfying a Wegner estimate for $\*$-boxes of the form \eqref{wesharp} in an open interval $\mathcal I$. Let us denote by $\blam$ be the  $n$-particle $\*$-box $\blam_{\*;L}^{(n)}$, where $\*\in\{\infty,S,H\}$.  Let $p_0>0$ and $\gamma>nd$.
  There exists a scale ${\mathcal L}={\mathcal L}(\gamma,n,d,\rho,p_0)$ such that, given  $E\in \mathcal I$, $L\geq \mathcal L$, and subsets $B_1, B_2\subset \blam$ such that $\partial_-\blam\subset B_2$,  for each $a>0$ and $\eps>0$ we have, for $\boldu\in B_1$ and $\y\in B_2$
\begin{align} \label{aGKlem6.4adisc}
& \P \pa{ a< \abs{G_{{\blam}}(E+i\eps;\boldy,\boldu)} } \notag \\
& \hskip 40pt \leq \frac{4L^{\gamma+2nd}}{a} \sup_{\boldk \in\partial_{+} {\blam} \cup B_2}\E\left( \abs{G(E+i\eps;\boldk,\boldu)}\right)+ p_0
\end{align}
and
\begin{align}\label{aGKlem6.4bdisc}
& \P \pa{ a< \abs{ G_{{\blam}}(E;\boldy,\boldu)} } \notag \\
& \hskip 20pt \leq  \frac{8L^{\gamma+2nd}}{a}  \sup_{\boldk\in\partial_{+} {\blam}\cup B_2}\E\left( \abs{G(E+i\eps;\boldk,\boldu)}\right)+  2^{3/2}C_n^{(\sharp)}\norm{\rho}_\infty\sqrt{\frac{\eps}{a}}L^{nd} +p_0.
\end{align}
\end{lemma}

\begin{proof}
Note that there exists a positive constant $C(n,d)$ such that for $L>C(n,d)$ we have $\abs{\partial_+\blam}\leq L^{nd}$ for all $\*\in\{\infty, S,H\}$, and center $\x\in\R^{nd}$.

We write $z:=E+i\eps\in\C$. We use the geometric resolvent identity (see \cite[Section 5]{Ki}) to get
\begin{align} G_{{\blam}}(z;\boldy,\boldu)= G(z;\boldy,\boldu) -\sum_{({\mathbf{k}^{\pr}},\boldk)\in \partial{\blam}} G(z;\boldk,{\boldu})G_{{\blam}}(z;\boldy, \mathbf{k}^{\pr})
\end{align}
where $(\mathbf{k}^{\pr},\boldk)\in \partial{\blam}$ means $\mathbf{k}^{\pr}\in\partial_-\blam$ and $\boldk\in\partial_+\blam$, with $\norm{\boldk-\boldk^\pr}_1=1$.

From this, we obtain
\be \abs{G_{{\blam}}(z;\boldy,\boldu)} \leq \abs{G(z;\boldy,\boldu)}+\norm{G_{{\blam}}(E)}\sum_{\boldk \in\partial_{+}{\blam}} \abs{G(z; \boldk,\boldu)}.\ee
Then,
\begin{align}\label{prob1}
& \P\left( a< \abs{G_{{\blam}}(z;\boldy,\boldu)} \right) \leq \P\left( \frac{a}{2}< \abs{G(z;\boldy,\boldu)} \right) \\
& \hskip 120pt + \P \left( \frac{a}{2}< \norm{G_{{\blam}}(E)}\sum_{\boldk \in\partial_{+}{\blam}} \abs{G(z;\boldk,\boldu)}\right). \notag
\end{align}
We bound the second term in the r.h.s. as follows,
\begin{align}
& \P \left( \frac{a}{2}< \norm{G_{{\blam}}(E)}\sum_{\boldk \in\partial_{+}{\blam}} \abs{G(z;\boldk,\boldu)}\right) \notag\\
&\hskip 10pt \le \P \left( \frac{a}{2L^\gamma}< \sum_{ \boldk  \in\partial_{+}{\blam}} \abs{G(z; \boldk  ,\boldu)}\right)+ \P \left( L^\gamma< \norm{G_{{\blam}}(E)}\right).
\end{align}
Note that
\be\label{probsum}\P \left( \frac{a}{2L^\gamma}< \sum_{ \boldk  \in\partial_{+}{\blam}} \abs{G(z; \boldk  ,\boldu)}\right)
\leq \abs{\partial_+\blam} \sup_{\boldk\in\partial_+\blam}
\P \left( \frac{a}{2L^\gamma \abs{\partial_+\blam} }< \abs{G(z; \boldk  ,\boldu)}\right).
\ee
We use this, Chebyshev's inequality and the Wegner estimate \eqref{wesharp} to bound \eqref{prob1} and obtain
\begin{align}\label{prob2}
& \P\left( a< \abs{G_{{\blam}}(z;\boldy,\boldu)} \right)\notag \\
& \leq \frac{2}{a} \E \left(\abs{G(z;\boldy,\boldu)}\right) +
\frac{2L^\gamma\abs{\partial_+\blam}^2}{a} \sup_{\boldk\in\partial_+\blam}\E \pa{\abs{G(z; \boldk  ,\boldu)} }\notag\\
& \quad\quad\quad + 2{C_n^{(\sharp)}}\norm{\rho}_\infty L^{-\gamma+nd}
\end{align}
We will take $L>\mathcal L_1$ for a $\mathcal L_1=\mathcal L_1(n,d)$ such that $\abs{\partial_+\blam}<L^{nd}$. We will, furthermore, take $L>\mathcal L_2=\mathcal L_2(\gamma,n,d,\rho,p_0)$ such that the last term in
\eqref{prob2} $2{C_n^{(\sharp)}}\norm{\rho}_\infty L^{-\gamma+nd}<p_0$.
\begin{align}
& \P\left( a< \abs{G_{{\blam}}(z;\boldy,\boldu)} \right) \\
& \leq \frac{2}{a} \E \left(\abs{G(z;\boldy,\boldu)}\right) +
\frac{2L^{\gamma+2nd}}{a} \sup_{\boldk\in\partial_+\blam}\E \pa{\abs{G(z; \boldk  ,\boldu)} }+ p_0.
\end{align}
We will give a common bound for the first two terms in the r.h.s., which yields, for $L\geq\mathcal L_0=\max\{\mathcal L_1,\mathcal L_2,1\}$ depending on $\gamma,n,d,\rho$ and $p_0$,
\be
\P\left( a< \abs{G_{{\blam}}(z;\boldy,\boldu)} \right)
\leq \frac{{4}L^{\gamma+2nd}}{a} \sup_{k\in\partial_+\blam\cup B_2}\E \pa{\abs{G(z; \boldk  ,\boldu)} }+ p_0.
\ee

Next, the resolvent identity gives
\begin{align}
G_{{\blam}}(E;\boldy,\boldu)& = G_{{\blam}}(E+i\eps;\boldy,\boldu)+i\eps G_{{\blam}}(E)G_{{\blam}}(E+i\eps)(\boldy,\boldu),
\end{align}
thus,
\be
\abs{ G_{{\blam}}(E;\boldy,\boldu)} \leq \abs{G_{\blam}(E+i\eps;\boldy,\boldu)}+\eps \norm{G_{\blam}(E)}^2. \ee
Finally,
\begin{align}
& \P\left( a< \abs{ G_{{\blam}}(E;\boldy,\boldu)} \right)  \notag  \\
&\hskip 15pt \leq \P\left( \frac{a}{2}< \abs{G_{\blam}(E+i\eps;\boldy,\boldu)}\right) + \P \left(\frac{a}{2\eps} <  \norm{G_{\blam}(E)}^2 \right)\\
& \hskip 30pt \leq \frac{{ 8}L^{\gamma+2nd}}{a} \sup_{\boldk\in\partial_{+} {\blam}\cup B_2}\E\left( \abs{G(E+i\eps;\boldk,\boldu)}\right)+ 2^{3/2}{ C_n^{(\sharp)}}\norm{\rho}_\infty\sqrt{\frac{\eps}{a}}L^{nd} +p_0. \notag
\end{align}
\end{proof}

\end{document}